\documentclass[]{article}
\usepackage{amsmath,amssymb,amsfonts}
\usepackage{algorithmic}
\usepackage{graphicx}
\usepackage{textcomp}
\def\BibTeX{{\rm B\kern-.05em{\sc i\kern-.025em b}\kern-.08em
    T\kern-.1667em\lower.7ex\hbox{E}\kern-.125emX}}
\usepackage{float}
\usepackage{comment}
    
\usepackage{booktabs} 
\usepackage{bm}
\usepackage{natbib}
\setcitestyle{comma}
\usepackage{amsthm}
\usepackage{stmaryrd}
\usepackage{setspace}
\usepackage{multirow}
\usepackage{subcaption}
\usepackage{authblk}

\usepackage{hyperref}
\usepackage[a4paper,%
            left=2cm,right=2cm]{geometry}

\hypersetup{
	colorlinks=true,
	urlcolor=blue,
	citecolor=blue,
	linktoc=all,
	linkcolor=blue}

  \theoremstyle{plain}
  \newtheorem{lyxalgorithm}{\protect\algorithmname}
\theoremstyle{plain}
\newtheorem{thm}{\protect\theoremname}
  \theoremstyle{plain}
  \newtheorem{prop}{\protect\propositionname}
  \theoremstyle{plain}
  \newtheorem{lem}{\protect\lemmaname}
  \theoremstyle{plain}
  \newtheorem{cor}{\protect\corollaryname}

\makeatother

\def\ddr  {{\rm d}}

  \providecommand{\algorithmname}{Algorithm}
  \providecommand{\lemmaname}{Lemma}
  \providecommand{\propositionname}{Proposition}
\providecommand{\corollaryname}{Corollary}
\providecommand{\theoremname}{Theorem}
    
\title{Approximate Bayesian computation via the energy statistic}
\author{
\uppercase{Hien D. Nguyen}$^{1}$,
\uppercase{Julyan Arbel}$^{2}$, \\
\uppercase{Hongliang L\"u}$^{2}$, 
\uppercase{Florence Forbes}$^{2}$
}

\affil{$^1$Department of Mathematics and Statistics, La Trobe University, Bundoora Melbourne 3066, Victoria Australia. (e-mail: h.nguyen5@latrobe.edu.au) $^2$Univ. Grenoble Alpes, Inria, CNRS, Grenoble INP, LJK, 38000 Grenoble, France}

\begin{document}

\maketitle

\begin{abstract}
Approximate Bayesian computation (ABC) has become an essential part of the Bayesian toolbox for addressing problems in which the likelihood is prohibitively expensive or entirely unknown, making it intractable. 
ABC defines a pseudo-posterior by comparing observed data with simulated data, traditionally based on some summary statistics, the elicitation of which is regarded as a key difficulty.  
Recently, using data discrepancy measures has been proposed in order to bypass the construction of summary statistics. 
Here we propose to use the importance-sampling ABC (IS-ABC) algorithm relying on the so-called \textit{two-sample energy statistic}. 
We establish a new asymptotic result for the case where both the observed sample size and the simulated data sample size increase to infinity, which highlights to what extent the data discrepancy measure impacts the asymptotic pseudo-posterior. 
The result holds in the broad setting of IS-ABC methodologies, thus generalizing previous results that have been established only for rejection ABC algorithms. 
Furthermore, we propose a consistent V-statistic estimator of the energy statistic, under which we show that the large sample result holds, and prove that the rejection ABC algorithm, based on the energy statistic, generates pseudo-posterior distributions that achieves  convergence to the correct limits, when implemented with rejection thresholds that converge to zero, in the finite sample setting.
Our proposed energy statistic based ABC algorithm is demonstrated on a variety of models, including a Gaussian mixture, a moving-average model of order two, a bivariate beta and a multivariate $g$-and-$k$ distribution. We find that our proposed method compares well with alternative discrepancy measures.
\end{abstract}

\section{Introduction}

In recent years, Bayesian inference has become a popular paradigm
for machine learning and statistical analysis. Good introductions
and references to the primary methods and philosophies of Bayesian
inference can be found in texts such as \cite{Press2003,Ghosh2006,Koch2007,Koop2007,Robert2007,Barber2012,Murphy2012}. 
When conducting parametric Bayesian inference, we observe some realization
$\bm{x}$ of the data $\bm{X}\in\mathbb{X}$ that are generated from
some data generating process (DGP), which can be characterized by
a parametric likelihood, given by a probability density function (PDF)
$f\left(\bm{x}|\bm{\theta}\right)$, determined entirely via the parameter
vector $\bm{\theta}\in\mathbb{T}$. 
Using expert knowledge, or based on computational considerations such as conjugacy, we endow the parameter $\bm{\theta}$ with some prior
PDF $\pi\left(\bm{\theta}\right)$. The goal of Bayesian inference is then to characterize the posterior
distribution
\begin{equation}
\pi\left(\bm{\theta}|\bm{x}\right)=\frac{f\left(\bm{x}|\bm{\theta}\right)\pi\left(\bm{\theta}\right)}{c\left(\bm{x}\right)}\text{,}\label{eq: Posterior Dist}
\end{equation}
where the prior predictive distribution $c\left(\bm{x}\right)$ is defined by
\[
c\left(\bm{x}\right)=\int_{\mathbb{T}}f\left(\bm{x}|\bm{\theta}\right)\pi\left(\bm{\theta}\right)\ddr  \bm{\theta}\; .
\]

In very simple cases, such as cases when the prior PDF is a conjugate
of the likelihood (cf. Sec. 3.3 of \cite{Robert2007}), the posterior
PDF (\ref{eq: Posterior Dist}) can be expressed explicitly.
In the case of more complex but still tractable pairs of likelihood
and prior PDFs, one can sample from (\ref{eq: Posterior Dist}) via
a variety of Monte Carlo methods, such as those reported in Ch. 6 of \cite{Press2003}.

In cases where the likelihood function is known but not tractable,
or when the likelihood function has entirely unknown form, one cannot
exactly sample from (\ref{eq: Posterior Dist}) in an inexpensive
manner, or at all. In such situations, a sample from an approximation
of (\ref{eq: Posterior Dist}) may suffice in order to conduct the
user's desired inference. Such a sample can be drawn via the method
of approximate Bayesian computation (ABC).

It is generally agreed that the ABC paradigm originated from the works
of \cite{Rubin1984,Pritchard1999}; see \cite{Tavare:2019aa} for details.
Stemming from the initial listed works, there are now numerous variants
of ABC methods. Some good reviews of the current ABC literature can
be found in the expositions of \cite{Marin2012,Voss2014,Lintusaari2017,Karabatsos2018}. The volume \cite{Sisson:2019aa} provides a comprehensive treatment regarding ABC methodologies.

The core philosophy of ABC is to define a pseudo-posterior by comparing data with plausibly simulated replicates. The comparison is traditionally based on some summary statistics, the choice of which being regarded as a key challenge of the approach.  

In recent years, data discrepancy measures bypassing the construction of summary statistics have been proposed by viewing data sets as empirical measures. \textcolor{black}{Recent examples of such an approach include the use of the maximum mean discrepancy (MMD) \citep{Park2016}, Kullback--Leibler divergence \citep{Jiang2018}, and the Wasserstein distance \citep{Bernton:2017aa}. Furthermore, \cite{Jiang2018} also considered the use of the classification accuracy method of \cite{gutmann2018likelihood}, and the indirect inference method of \cite{drovandi2015bayesian}, in the data discrepancy context.}

In this article, we develop upon the discrepancy measurement approach
of \cite{Jiang2018}, via the importance sampling ABC (IS-ABC) approach, which makes use of a weight function; see e.g. \cite{Karabatsos2018}. In particular, we report on a class of
ABC algorithms that utilize the two-sample energy statistic (ES) of
\cite{Szekely2004} (see also \cite{Baringhaus2004,Szekely2013,Szekely2017,mak2018support}). Our approach is related to the MMD ABC algorithms that were implemented in \cite{Park2016,Jiang2018,Bernton:2017aa}. The MMD is a discrepancy
measurement that is closely related to the ES, cf. \cite{Sejdinovic:2013aa}.

We establish new asymptotic results that have not
been proved in these previous papers. In the IS-ABC setting and in the regime where both the observation sample size and the simulated data sample size increase to infinity, our theoretical result highlights how the data discrepancy measure impacts the asymptotic pseudo-posterior. 
More specifically, we make the assumption that the data discrepancy measure converges to some asymptotic value $\mathcal{D}_{\infty}\left(\bm{\theta}_{0},\bm{\theta}\right)$,
where $\bm{\theta}_{0}$ stands for the `true' parameter value associated to the DGP that generates observations $\bm{X}$. 
We then show that the pseudo-posterior distribution converges almost surely to a distribution depending on the prior $\pi$ and on the limiting value $\mathcal{D}_{\infty}\left(\bm{\theta}_{0},\bm{\theta}\right)$. 
In addition to our asymptotic results regarding large sample scenarios, we also provide corollaries regarding the performance of our ES-based ABC method, due to the general finite sample theoretical results of \cite{Bernton:2017aa}. \textcolor{black}{Our asymptotic results provide useful approximations and guarantees for the practical application of our method.}

The last decade has seen an active development in research on asymptotic properties of ABC. Early works revolved around the impact of the acceptance threshold on the ABC bias and the Monte Carlo error \citep{blum2010approximate,barber2015rate,biau2015new}, and on the choice of summary statistics \citep{blum2010approximate,fearnhead2012constructing,prangle2014semi}.  
Further works focused on large sample size properties such as 
consistency for model choice \citep{marin2014relevant}, 
asymptotic efficiency \citep{li2018asymptotic}, posterior consistency, and contraction rates \citep{frazier2018asymptotic}. It is with these results, where our article fits.
Although devised in settings where likelihoods are assumed intractable, ABC can also be cast in the setting of robustness with respect to misspecification \citep{frazier2020model}. In particular,  the ABC posterior distribution can be viewed as a special case of a coarsened posterior distribution \citep{miller2018robust}.

The remainder of the article proceeds as follows. In Section~\ref{sec:is-abc}, we introduce
the general IS-ABC framework. In Section~\ref{sec:energy}, we introduce the two-sample
ES and demonstrate how it can be incorporated into the IS-ABC framework.
Theoretical results regarding the IS-ABC framework and the two-sample
ES are presented in Section~\ref{sec:theory}. Illustrations of the IS-ABC framework are presented in Section~\ref{sec:illustrations}. Conclusions are drawn in Section~\ref{sec:conclusion}. 

\section{Importance sampling ABC} \label{sec:is-abc}

Assume that we observe $n$ independent and identically distributed
(IID) replicates of $\bm{X}$ from some DGP, which we put into $\mathbf{X}_{n}=\left\{ \bm{X}_{i}\right\} _{i=1}^{n}$.
We suppose that the DGP that generates $\bm{X}$ is dependent on some
parameter vector $\bm{\theta}$ from space $\mathbb{T}$,
which is random and has prior PDF $\pi\left(\bm{\theta}\right)$.

Denote $f\left(\bm{x}|\bm{\theta}\right)$ to be the PDF of $\bm{X}$,
given $\bm{\theta}$, and write
\[
f\left(\mathbf{x}_{n}|\bm{\theta}\right)=\prod_{i=1}^{n}f\left(\bm{x}_{i}|\bm{\theta}\right)\text{,}
\]
where $\mathbf{x}_{n}$ is a realization of $\mathbf{X}_{n}$, and
each $\bm{x}_{i}$ is a realization of $\bm{X}_{i}$ ($i\in\left[n\right]=\left\{ 1,\dots,n\right\} $).

If $f\left(\mathbf{x}_{n}|\bm{\theta}\right)$ were known, then we
could use (\ref{eq: Posterior Dist}) to write the posterior PDF
\begin{equation}
\pi\left(\bm{\theta}|\mathbf{x}_{n}\right)=\frac{f\left(\mathbf{x}_{n}|\bm{\theta}\right)\pi\left(\bm{\theta}\right)}{c\left(\mathbf{x}_{n}\right)}\text{,}\label{eq: posterior}
\end{equation}
where $c\left(\mathbf{x}_{n}\right)=\int_{\mathbb{T}} f\left(\mathbf{x}_{n}|\bm{\theta}\right)\pi\left(\bm{\theta}\right)\ddr \bm{\theta}$
is a constant that makes $\int_{\mathbb{T}}\pi\left(\bm{\theta}|\mathbf{x}_{n}\right)\ddr \bm{\theta}=1$.
When evaluating $f\left(\bm{x}|\bm{\theta}\right)$ is prohibitive and ABC is required,
then operating with $f\left(\mathbf{x}_{n}|\bm{\theta}\right)$ is
similarly difficult. We suppose that given any $\bm{\theta}\in\mathbb{T}$,
we at least have the capability of sampling from the DGP with PDF
$f\left(\bm{x}|\bm{\theta}\right)$.  
That is, we have a simulation
method that allows us to feasibly sample the IID vector $\mathbf{Y}_{m}=\left\{ \bm{Y}_{i}\right\} _{i=1}^{m}$,
for any $m\in\mathbb{N}$, for a DGP with PDF
\[
f\left(\mathbf{y}_{n}|\bm{\theta}\right)=\prod_{i=1}^{m}f\left(\bm{y}_{i}|\bm{\theta}\right)\text{.}
\]
\textcolor{black}{Typically, one should choose $m=n$, as it fulfils the hypotheses of all of our proved theoretical results. This choice is made throughout all of our numerical demonstrations. However, we anticipate that there may be practical or computational scenarios, where it may be advantageous to be able to choose $m\ne n$, which is permissible in our methodological framework.}

Using the simulation mechanism that generates samples $\mathbf{Y}_{m}$
and the prior distribution that generates parameters $\bm{\theta}$,
we can simulate a set of $N\in\mathbb{N}$ simulations $\mathbf{Z}_{N}=\left\{ \bm{Z}_{m,k}\right\} _{k=1}^{N}$,
where $\bm{Z}_{m,k}^{\top}=\left(\mathbf{Y}_{m,k}^{\top},\bm{\theta}_{k}^{\top}\right)$
and $\left(\cdot\right)^{\top}$ is the transposition operator. Here,
for each $k\in\left[N\right]$, $\bm{Z}_{m,k}$ is an observation from the DGP with joint PDF $f\left(\mathbf{y}_{m}|\bm{\theta}\right)\pi\left(\bm{\theta}\right)$, hence each $\bm{Z}_{m,k}$ is  composed of a parameter value and a datum conditional on the parameter value. 
We now consider how $\mathbf{X}_{n}$ and $\mathbf{Z}_{N}$ can be
combined in order to construct an approximation of (\ref{eq: posterior}).

Following the approach of \cite{Jiang2018}, we define $\mathcal{D}\left(\mathbf{x}_{n},\mathbf{y}_{m}\right)$
to be some non-negative real-valued function that outputs a small
value if $\mathbf{x}_{n}$ and $\mathbf{y}_{m}$ are similar, and
outputs a large value if $\mathbf{x}_{n}$ and $\mathbf{y}_{m}$ are
different, in some sense. We call $\mathcal{D}\left(\mathbf{x}_{n},\mathbf{y}_{m}\right)$
the data discrepancy measurement between $\mathbf{x}_{n}$ and $\mathbf{y}_{m}$,
and we say that $\mathcal{D}\left(\cdot,\cdot\right)$ is the data discrepancy
function.

Next, we let $w\left(d,\epsilon\right)$ be a non-negative, decreasing (in $d$), and bounded
(importance sampling) weight function (cf. Section 3 of \cite{Karabatsos2018}),
which takes as inputs a data discrepancy measurement $d=\mathcal{D}\left(\mathbf{x}_{n},\mathbf{y}_{m}\right)\ge0$ and
a calibration parameter $\epsilon>0$. Using the weight and discrepancy
functions, we can propose the following approximation for (\ref{eq: posterior}).

In the language of \cite{Jiang2018}, we call
\begin{equation}
\pi_{m,\epsilon}\left(\bm{\theta}|\mathbf{x}_{n}\right)=\frac{\pi\left(\bm{\theta}\right)L_{m,\epsilon}\left(\mathbf{x}_{n}|\bm{\theta}\right)}{c_{m,\epsilon}\left(\mathbf{x}_{n}\right)}\label{eq: pseudo posterior}
\end{equation}
the pseudo-posterior PDF, where
\[
L_{m,\epsilon}\left(\mathbf{x}_{n}|\bm{\theta}\right)=\int_{\mathbb{X}^{m}}w\left(\mathcal{D}\left(\mathbf{x}_{n},\mathbf{y}_{m}\right),\epsilon\right)f\left(\mathbf{y}_{m}|\bm{\theta}\right)\ddr \mathbf{y}_{m}
\]
is the approximate likelihood function, and

\[
c_{m,\epsilon}\left(\mathbf{x}_{n}\right)=\int_{\mathbb{T}}\pi\left(\bm{\theta}\right)L_{m,\epsilon}\left(\mathbf{x}_{n}|\bm{\theta}\right)\ddr \bm{\theta}
\]
is a normalization constant. We can use (\ref{eq: pseudo posterior})
to approximate (\ref{eq: posterior}) in the following way. For any
functional of the parameter vector $\bm{\theta}$ of interest, $g\left(\bm{\theta}\right)$ say,
we may approximate the posterior mean Bayesian estimator of $g\left(\bm{\theta}\right)$
via the expression

\begin{equation}
\mathbb{E}\left[g\left(\bm{\theta}\right)|\mathbf{x}_{n}\right]\approx\frac{\int_{\mathbb{T}}g\left(\bm{\theta}\right)\pi\left(\bm{\theta}\right)L_{m,\epsilon}\left(\mathbf{x}_{n}|\bm{\theta}\right)\ddr \bm{\theta}}{c_{m,\epsilon}\left(\mathbf{x}_{n}\right)}\text{,}\label{eq: approx Bayes est}
\end{equation}
where the right-hand side of (\ref{eq: approx Bayes est}) can be
unbiasedly estimated using $\mathbf{Z}_{N}$ via
\begin{equation}
\mathbb{M}\left[g\left(\bm{\theta}\right)|\mathbf{x}_{n}\right]=\frac{\sum_{k=1}^{N}g\left(\bm{\theta}_{k}\right)w\left(\mathcal{D}\left(\mathbf{x}_{n},\mathbf{Y}_{m,k}\right),\epsilon\right)}{\sum_{k=1}^{N}w\left(\mathcal{D}\left(\mathbf{x}_{n},\mathbf{Y}_{m,k}\right),\epsilon\right)}\text{.}\label{eq: Sample approx Bayes est}
\end{equation}

We call the process of constructing (\ref{eq: Sample approx Bayes est}),
to approximate (\ref{eq: approx Bayes est}), the IS-ABC procedure.
The general form of the IS-ABC procedure is provided in Algorithm
\ref{alg:IS-ABC}.
\begin{lyxalgorithm}
\label{alg:IS-ABC}IS-ABC procedure for approximating $\mathbb{E}\left[g\left(\bm{\theta}\right)|\mathbf{x}_{n}\right]$.

\textbf{Input:} a data discrepancy function $\mathcal{D}$, a weight function
$w$, and a calibration parameter $\epsilon>0$.

\textbf{For} $k\in\left[N\right]$;

sample $\bm{\theta}_{k}$ from PDF $\pi\left(\bm{\theta}\right)$;

generate $\mathbf{Y}_{m,k}$ from the DGP with PDF $f\left(\mathbf{y}_{m}|\bm{\theta}_{k}\right)$;

put $\bm{Z}_{k}=\left(\mathbf{Y}_{m,k},\bm{\theta}_{k}\right)$ into
$\mathbf{Z}_{N}$.

\textbf{Output:} $\mathbf{Z}_{N}$ and construct the estimator $\mathbb{M}\left[g\left(\bm{\theta}\right)|\mathbf{x}_{n}\right]$.
\end{lyxalgorithm}

\section{The energy statistic (ES)} \label{sec:energy}

Let $\delta$ define a metric and let $\bm{X}\in\mathbb{X}\subseteq\mathbb{R}^{d}$
and $\bm{Y}\in\mathbb{X}$ be two random variables that are in a space endowed with a semi-metric $\delta$, where $d\in\mathbb{N}$  (cf. \cite{Sejdinovic:2013aa}). 
Furthermore, let $\bm{X}^{\prime}$ and $\bm{Y}^{\prime}$ be two
random variables that have the same distributions as $\bm{X}$ and
$\bm{Y}$, respectively. Here, $\bm{X}$, $\bm{X}^{\prime}$, $\bm{Y}$,
and $\bm{Y}^{\prime}$ are all independent of one another.

Upon writing
\[
\mathcal{E}_{\delta}\left(\bm{X},\bm{Y}\right)=2\mathbb{E}\left[\delta\left(\bm{X},\bm{Y}\right)\right]-\mathbb{E}\left[\delta\left(\bm{X},\bm{X}^{\prime}\right)\right]-\mathbb{E}\left[\delta\left(\bm{Y},\bm{Y}^{\prime}\right)\right]\text{,}
\]
we can define the original ES of \cite{Baringhaus2004} and \cite{Szekely2004},
as a function of $\bm{X}$ and $\bm{Y}$, via the expression $\mathcal{E}_{\delta_{1}}\left(\bm{X},\bm{Y}\right)$,
where $\delta_{\beta}\left(\bm{x},\bm{y}\right)=\left\Vert \bm{x}-\bm{y}\right\Vert _{2}^{\beta}$
is the $\beta$ power of the metric corresponding to the $L_{2}\text{-norm}$ ($\beta\in\left(0, 2\right]$; cf. \cite[Prop. 2]{Szekely2013}).  
Thus, the original ES statistic, which we shall also denote as $\mathcal{E}\left(\bm{X},\bm{Y}\right)$,
is defined using the Euclidean metric $\delta_1$.

The original ES has numerous useful mathematical properties. For instance,
under the assumption that $\mathbb{E}\left\Vert \bm{X}\right\Vert _{2}+\mathbb{E}\left\Vert \bm{Y}\right\Vert _{2}<\infty$,
it was shown that
\begin{equation}
\mathcal{E}\left(\bm{X},\bm{Y}\right)=\frac{\Gamma\left(\frac{d+1}{2}\right)}{\pi^{\left(d+1\right)/2}}\int_{\mathbb{R}^{d}}\frac{\left|\varphi_{X}\left(\bm{t}\right)-\varphi_{Y}\left(\bm{t}\right)\right|^{2}}{\left\Vert \bm{t}\right\Vert _{2}^{d+1}}\ddr \bm{t}\text{,}\label{eq: ES limit}
\end{equation}
in Proposition 1 of \cite{Szekely2013}, where $\Gamma\left(\cdot\right)$
is the gamma function and $\varphi_{X}$ (respectively, $\varphi_{Y}$) is the characteristic function of $\bm{X}$ (respectively, $\bm{Y}$). Thus, we have the fact that $\mathcal{E}\left(\bm{X},\bm{Y}\right)\ge0$
for any $\bm{X},\bm{Y}\in\mathbb{X}$, and $\mathcal{E}\left(\bm{X},\bm{Y}\right)=0$
if and only if $\bm{X}$ and $\bm{Y}$ are identically distributed. 

The result above is generalized in Proposition 3 of \cite{Szekely2013},
where we have the following statement. If $\delta\left(\bm{x},\bm{y}\right)=\delta\left(\bm{x}-\bm{y}\right)$
is a continuous function and $\bm{X},\bm{Y}\in\mathbb{R}^{d}$ are
independent random variables, then it is necessary and sufficient
that $\delta\left(\cdot\right)$ is strictly negative definite (see \cite{Szekely2013} for the precise definition) for
the following conclusion to hold: $\mathcal{E}_{\delta}\left(\bm{X},\bm{Y}\right)\ge0$
for any $\bm{X},\bm{Y}\in\mathbb{X}$, and $\mathcal{E}_{\delta}\left(\bm{X},\bm{Y}\right)=0$
if and only if $\bm{X}$ and $\bm{Y}$ are identically distributed.

We observe that there is thus an infinite variety of functions
$\delta$ from which we can construct energy statistics. We shall
concentrate on the use of the original ES, based on $\delta_{1}$,
since it is the most well known and popular of the varieties.

\subsection{The V-statistic estimator}

Suppose that we observe $\mathbf{X}_{n}=\left\{ \bm{X}_{i}\right\} _{i=1}^{n}$
and $\mathbf{Y}_{m}=\left\{ \bm{Y}_{i}\right\} _{i=1}^{m}$, where
the former is a sample containing $n$ IID replicates of $\bm{X}$, and the latter
is a sample containing $m$ IID replicates of $\bm{Y}$, respectively, with $\mathbf{X}_{n}$
and $\mathbf{Y}_{m}$ being independent. In \cite{Gretton:2012aa},
it was shown that for any $\delta$, upon assuming that $\delta\left(\bm{x},\bm{y}\right)<\infty$,
the so-called V-statistic estimator (cf. \cite[Ch. 5]{Serfling:1980aa}
and \cite{Koroljuk:1994aa})
\begin{align}
\mathcal{V}_{\delta}\left(\mathbf{X}_{n},\mathbf{Y}_{m}\right)=&\frac{2}{mn}\sum_{i=1}^{n}\sum_{j=1}^{m}\delta\left(\bm{X}_{i},\bm{Y}_{j}\right) \nonumber \\
-&\frac{1}{n^{2}}\sum_{i=1}^{n}\sum_{j=1}^{n}\delta\left(\bm{X}_{i},\bm{X}_{j}\right) \nonumber \\
-&\frac{1}{m^{2}}\sum_{i=1}^{m}\sum_{j=1}^{m}\delta\left(\bm{Y}_{i},\bm{Y}_{j}\right)\text{,} \label{eq: V-stat}
\end{align} 
can be proved to converge in probability to $\mathcal{E}_{\delta}\left(\bm{X},\bm{Y}\right)$,
as $n\rightarrow\infty$ and $m\rightarrow\infty$, under the condition
that $m/n\rightarrow\alpha<\infty$, for some constant $\alpha$ (see
also \cite{Gretton:2007aa}). Here, the proof was provided in the context of MMDs (see definition in Section \ref{sec:rel}) but is easily portable to the ES setting. 

We note that the assumption of this
result is rather restrictive, since it either requires the bounding
of the space $\mathbb{X}$ or the function $\delta$. In the sequel,
we will present a result for the almost sure convergence of the V-statistic
that depends on the satisfaction of a more realistic hypothesis.

It is noteworthy that if the ES is non-negative, then the V-statistic
retains the non-negativity property of its corresponding ES (cf. \cite{Gretton:2012aa}).
That is, for any continuous and negative definite function $\delta\left(\bm{x},\bm{y}\right)=\delta\left(\bm{x}-\bm{y}\right)$,
we have $\mathcal{V}_{\delta}\left(\mathbf{X}_{n},\mathbf{Y}_{m}\right)\ge0$.

\subsection{The ES-based IS-ABC algorithm}

From Algorithm \ref{alg:IS-ABC}, we observe that an IS-ABC algorithm
requires three components. A data discrepancy measurement $d=\mathcal{D}\left(\mathbf{X}_{n},\mathbf{Y}_{m}\right)\ge0$,
a weighting function $w\left(d,\epsilon\right)\ge0$, and a tuning
parameter $\epsilon>0$. We propose the use of the ES in the place
of the data discrepancy measurement $d$, in combination with various
weight functions that have been used in the literature. That is we
set

\[
\mathcal{D}\left(\mathbf{X}_{n},\mathbf{Y}_{m}\right)=\mathcal{V}_{\delta}\left(\mathbf{X}_{n},\mathbf{Y}_{m}\right)\text{,}
\]
in Algorithm \ref{alg:IS-ABC}. 

In particular, we consider original ES, where $\delta=\delta_{1}$.
We name our framework the ES-ABC algorithm. In Section~\ref{sec:theory}, we shall
demonstrate that the proposed algorithm possesses desirable large sample
qualities that guarantees its performance in practice, as illustrated in Section~\ref{sec:illustrations}. 

\subsection{Related methods}
\label{sec:rel}

The ES-ABC algorithm that we have presented here is closely
related to ABC algorithms based on the maximum mean discrepancy (MMD) that were implemented in \cite{Park2016}, \cite{Jiang2018}, and \cite{Bernton:2017aa}. For each Mercer kernel function $\chi\left(\bm{x},\bm{y}\right)$
($\bm{x},\bm{y}\in\mathbb{X})$, the corresponding MMD is defined
via the equation
\begin{align*}
\text{MMD}_{\chi}^{2}\left(\bm{X},\bm{Y}\right)=&\mathbb{E}\left[\chi\left(\bm{X},\bm{X}^{\prime}\right)\right]+\mathbb{E}\left[\chi\left(\bm{Y},\bm{Y}^{\prime}\right)\right]\\
-&2\mathbb{E}\left[\chi\left(\bm{X},\bm{Y}\right)\right]\text{,}
\end{align*}
where $\bm{X},\bm{X}^{\prime},\bm{Y},\bm{Y}^{\prime}$ are random
variable such that $\bm{X}$ and $\bm{Y}$ are identically distributed
to $\bm{X}^{\prime}$ and $\bm{Y}^{\prime}$, respectively. 

The MMD as a statistic for testing goodness-of-fit was studied prominently
in articles such as \cite{Gretton:2007aa}, \cite{Gretton:2009aa},
and \cite{Gretton:2012aa}.  More details
regarding the relationship between the two classes of statistics can
be found in \cite{Sejdinovic:2013aa}.

We note two shortcomings with respect to the applications of the MMD
as a basis for an ABC algorithm in the previous literature. Firstly,
no theoretical results regarding the consistency of the MMD-based
methods have been proved. And secondly, in the application by \cite{Park2016}
and \cite{Jiang2018}, the MMD was implemented using the unbiased
U-statistic estimator, rather than the biased V-statistic estimator.
Although both estimators are consistent, in the sense that they can
be proved to be convergent to the desired limiting MMD value, the
U-statistic estimator has the property of not being bounded
from below by zero (cf. \cite{Gretton:2012aa}). As such, it does
not meet the strict definition of a data discrepancy measurement.

For a sufficiently large sample size, the U-statistic will have low probability of having a value less than zero, and thus the difference between the U-statistic and V-statistic becomes immaterial for large $n$. One may also consider a truncation of the U-statistic, which causes no issues, asymptotically, as the U-statistic and V-statistic have the same limit, which is guaranteed to be non-negative. 

\section{Theoretical results } \label{sec:theory}

\subsection[Behavior as n and m go to infinity]{Behavior as $\MakeLowercase{n}\to\infty$ and $\MakeLowercase{m}\to\infty$}

\subsubsection{Analysis with a generic discrepancy}

We now establish a consistency result for the pseudo-posterior density
(\ref{eq: pseudo posterior}), when $n$ and $m$ approach infinity.
Our result generalizes the main result of \cite{Jiang2018} (i.e.,
Theorem 1), which is the specific case when the weight function is
restricted to the form

\begin{equation}
w\left(d,\epsilon\right)=\left\llbracket d<\epsilon\right\rrbracket \text{,}\label{eq: rejection}
\end{equation}
where $\left\llbracket \cdot\right\rrbracket $ is the Iverson bracket
notation, which equals 1 when the internal statement is true, and
0, otherwise (cf. \cite{Graham:1994aa}).

The weighting function of form (\ref{eq: rejection}), when implemented
within the IS-ABC framework, produces the common rejection ABC algorithms,
that were suggested by \cite{Tavare1997}, and \cite{Pritchard1999}.
We extended upon the result of \cite{Jiang2018} so that we may provide
theoretical guarantees for more exotic ABC procedures, such as the
kernel-smoothed ABC procedure of \cite{Park2016}, which implements
weights of the form
\begin{equation}
w\left(d,\epsilon\right)=\exp\left(-d^{q}/\epsilon\right)\text{,}\label{eq: kernel}
\end{equation}
for $q>0$. See \cite{Karabatsos2018} for further discussion and
examples.

In order to prove our asymptotic result, we require Hunt's lemma,
which is reported in \cite{Dellacherie1980}, as Theorem 45 of Section
V.5. For convenience to the reader, we present the result, below.
\begin{thm}
\label{thm: Hunt's lemma-1}Let $\left(\Omega,\mathcal{F},\mathbb{P}\right)$
be a probability space with increasing $\sigma\text{-fields}$ $\left\{ \mathcal{F}_{n}\right\} $
and let $\mathcal{F}_{\infty}=\cup_{n}\mathcal{F}_{n}$. Suppose that
$\left\{ U_{n}\right\} $ is a sequence of random variables that is
bounded from above in absolute value by some integrable random variable
$V$, and further suppose that $U_{n}$ converges almost surely to
the random variable $U$. Then, $\lim_{n\rightarrow\infty}\mathbb{E}\left(U_{n}|\mathcal{F}_{n}\right)=\mathbb{E}\left(U|\mathcal{F}_{\infty}\right)$ almost surely, and in  $\mathcal{L}_{1}$ mean, as $n\rightarrow\infty$.
\end{thm}
Define the continuity set of a function $d\mapsto w\left(d\right)$ as
\[
C\left(w\right)=\left\{ d:w\text{ is continuous at }d\right\} \text{.}
\]
Using Theorem \ref{thm: Hunt's lemma-1}, we can now prove the following
result regarding the asymptotic behavior of the pseudo-posterior density
function (\ref{eq: pseudo posterior}).
\begin{thm}
\label{thm: Asymptotic pseudo-posteriori-1}Let $\mathbf{X}_{n}$ and
$\mathbf{Y}_{m}$ be IID samples from DGPs that can be characterized
by PDFs $f\left(\mathbf{x}_{n}|\bm{\theta}_{0}\right)=\prod_{i=1}^{n}f\left(\bm{x}_{i}|\bm{\theta}_{0}\right)$
and $f\left(\mathbf{y}_{m}|\bm{\theta}\right)=\prod_{i=1}^{m}f\left(\bm{y}_{i}|\bm{\theta}\right)$,
respectively, with corresponding parameter vectors $\bm{\theta}_{0}$
and $\bm{\theta}$. Suppose that the data discrepancy $\mathcal{D}\left(\mathbf{X}_{n},\mathbf{Y}_{m}\right)$
converges to some $\mathcal{D}_{\infty}\left(\bm{\theta}_{0},\bm{\theta}\right)$,
which is a function of $\bm{\theta}_{0}$ and $\bm{\theta}$, almost
surely as $n\rightarrow\infty$, for some $m=m\left(n\right)\rightarrow\infty$.
If $w\left(d,\epsilon\right)$ is piecewise continuous and decreasing in $d$ and
$w\left(d,\epsilon\right)\le a<\infty$ for all $d\ge0$ and any $\epsilon>0$,
and if
\[
\mathcal{D}_{\infty}\left(\bm{\theta}_{0},\bm{\theta}\right)\in C\left(w\left(\cdot,\epsilon\right)\right)\text{,}
\]
then we have
\begin{equation}
\pi_{m,\epsilon}\left(\bm{\theta}|\mathbf{X}_{n}\right)\rightarrow\frac{\pi\left(\bm{\theta}\right)w\left(\mathcal{D}_{\infty}\left(\bm{\theta}_{0},\bm{\theta}\right),\epsilon\right)}{\int\pi\left(\bm{\theta}\right)w\left(\mathcal{D}_{\infty}\left(\bm{\theta}_{0},\bm{\theta}\right),\epsilon\right)\ddr \bm{\theta}} \text{,}\label{eq: convergence equation}
\end{equation}
almost surely, as $n\rightarrow\infty$.
\end{thm}
\begin{proof}
Using the notation of Theorem \ref{thm: Hunt's lemma-1}, we set $U_{n}=w\left(d\left(\mathbf{X}_{n},\mathbf{Y}_{m}\right),\epsilon\right)$.
Since $w\left(d,\epsilon\right)\le a<\infty$, for any $d$, we have
the existence of a $\left|U_{n}\right|\le V<\infty$ such that $V$
is integrable, since we can take $V=a$. Since $\mathcal{D}\left(\mathbf{X}_{n},\mathbf{Y}_{m}\right)$
converges almost surely to $\mathcal{D}_{\infty}\left(\bm{\theta}_{0},\bm{\theta}\right)$,
and $w\left(\cdot,\epsilon\right)$ is continuous at $\mathcal{D}_{\infty}\left(\bm{\theta}_{0},\bm{\theta}\right)$, we have $U_{n}\rightarrow U=w\left(\mathcal{D}_{\infty}\left(\bm{\theta}_{0},\bm{\theta}\right),\epsilon\right)$ with probability one 
by the extended continuous mapping theorem (cf. \cite[Thm. 7.10]{DasGupta2011}).

Now, let $\mathcal{F}_{n}$ be the $\sigma\text{-field}$ generated
by the sequence $\left\{ \bm{X}_{1},\dots,\bm{X}_{n}\right\} $. Thus,
$\mathcal{F}_{n}$ is an increasing $\sigma\text{-field}$, which
approaches $\mathcal{F}_{\infty}=\cup_{n}\mathcal{F}_{n}$. We are
in a position to directly apply Theorem \ref{thm: Hunt's lemma-1}.
This yields 
\[
\mathbb{E}\left[w\left(\mathcal{D}\left(\mathbf{X}_{n},\mathbf{Y}_{m}\right),\epsilon\right)|\textcolor{black}{\mathcal{F}_{n}}\right]\rightarrow\mathbb{E}\left[w\left(\mathcal{D}_{\infty}\left(\bm{\theta}_{0},\bm{\theta}\right),\epsilon\right)|\mathcal{F}_{\infty}\right]\text{,}
\]
almost surely, as $n\rightarrow\infty$, where the right-hand side
equals $w\left(\mathcal{D}_{\infty}\left(\bm{\theta}_{0},\bm{\theta}\right),\epsilon\right)$. 

Notice that the left-hand side has the form
\[
\mathbb{E}\left[w\left(\mathcal{D}\left(\mathbf{X}_{n},\mathbf{Y}_{m}\right),\epsilon\right)|\mathcal{F}_{n}\right]=L_{m,\epsilon}\left(\mathbf{X}_{n}|\bm{\theta}\right)
\]
and therefore $L_{m,\epsilon}\left(\mathbf{X}_{n}|\bm{\theta}\right)\rightarrow w\left(\mathcal{D}_{\infty}\left(\bm{\theta}_{0},\bm{\theta}\right),\epsilon\right)$,
almost surely, as $n\rightarrow\infty$. Thus, the numerator of (\ref{eq: pseudo posterior})
converges to
\begin{equation}
\pi\left(\bm{\theta}\right)w\left(\mathcal{D}_{\infty}\left(\bm{\theta}_{0},\bm{\theta}\right),\epsilon\right)\text{,}\label{eq: numerator1-1}
\end{equation}
almost surely. 

To complete the proof, it suffices to show that the denominator of
(\ref{eq: pseudo posterior}) converges almost surely to
\begin{equation}
\int_{\mathbb{T}}\pi\left(\bm{\theta}\right)w\left(\mathcal{D}_{\infty}\left(\bm{\theta}_{0},\bm{\theta}\right),\epsilon\right)\ddr \bm{\theta}\text{.}\label{eq: denominator1-1}
\end{equation}
Since $L_{m,\epsilon}\left(\mathbf{X}_{n}|\bm{\theta}\right)\rightarrow w\left(\mathcal{D}_{\infty}\left(\bm{\theta}_{0},\bm{\theta}\right),\epsilon\right)$
and $c_{m,\epsilon}\left(\mathbf{x}_{n}\right)=\int_{\mathbb{T}}\pi\left(\bm{\theta}\right)L_{m,\epsilon}\left(\mathbf{x}_{n}|\bm{\theta}\right)\ddr \bm{\theta}$,
we obtain our desired convergence via the dominated convergence theorem,
because $w\left(d,\epsilon\right)\le a<\infty$. An application of
a continuous mapping theorem (cf. \cite[Thm. 7.8]{DasGupta2011}) yields the almost sure convergence of the ratio
between (\ref{eq: numerator1-1}) and (\ref{eq: denominator1-1})
to the right-hand side of (\ref{eq: convergence equation}), as $n\rightarrow\infty$.
\end{proof}
The following result and proof guarantees the applicability of Theorem
\ref{thm: Asymptotic pseudo-posteriori-1} to rejection ABC procedures,
and to kernel-smoothed ABC procedures, as used in \cite{Jiang2018}
and \cite{Park2016}, respectively.
\begin{prop}
The result of Theorem \ref{thm: Asymptotic pseudo-posteriori-1} applies
to rejection ABC and importance sampling ABC, with weight functions of respective forms (\ref{eq: rejection}) and (\ref{eq: kernel}).
\end{prop}
\begin{proof}
For weights of form (\ref{eq: rejection}), we note that $w\left(d,\epsilon\right)=\left\llbracket d<\epsilon\right\rrbracket $
is continuous in $d$ at all points, other than when $d=\epsilon$.
Furthermore, $w\left(d,\epsilon\right)\in\left\{ 0,1\right\} $ and
is hence non-negative and bounded. Thus, under the condition that
$\mathcal{D}_{\infty}\left(\bm{\theta}_{0},\bm{\theta}\right)\ne\epsilon$,
we have the desired conclusion of Theorem \ref{thm: Asymptotic pseudo-posteriori-1}.

For weights of form (\ref{eq: kernel}), we note that for fixed $\epsilon$,
$w\left(d,\epsilon\right)$ is continuous and positive in $d$. Since $w$ is uniformly bounded by 1, 
differentiating with respect to $d$, we obtain $\ddr  w/\ddr  d=-\left(q/\epsilon\right)d^{q-1}\exp\left(-d^{q}/\epsilon\right)$, which is negative for any $d\ge0$ and $q>0$. Thus, 
(\ref{eq: kernel}) constitutes a weight function and
 satisfies the conditions of Theorem~\ref{thm: Asymptotic pseudo-posteriori-1}.
\end{proof}
%

\subsubsection{Analysis with the energy statistic}


We write $\log^{+}x=\log\left(\max\left\{ 1,x\right\} \right)$. From
\cite{Szekely2004} we have the fact that for arbitrary $\delta$,
\[
\mathcal{V}_{\delta}\left(\mathbf{X}_{n},\mathbf{Y}_{m}\right)=\sum_{i_{1}=1}^{n}\sum_{i_{2}=1}^{n}\sum_{j_{1}=1}^{m}\sum_{j_{2}=1}^{m}\frac{\kappa_{\delta}\left(\bm{X}_{i_{1},}\bm{X}_{i_{2}};\bm{Y}_{j_{1}},\bm{Y}_{j_{2}}\right)}{m^2 n^2}\text{,}
\]
where
\begin{align}
\kappa_{\delta}\left(\bm{x}_{i_{1}},\bm{x}_{i_{2}};\bm{y}_{j_{1}},\bm{y}_{j_{2}}\right)=&\delta\left(\bm{x}_{i_{1}},\bm{y}_{j_{1}}\right)+\delta\left(\bm{x}_{i_{2}},\bm{y}_{j_{2}}\right) \nonumber \\
-&\delta\left(\bm{x}_{i_{1}},\bm{x}_{i_{2}}\right)-\delta\left(\bm{y}_{j_{1}},\bm{y}_{j_{2}}\right) \nonumber
\end{align}
is the kernel of the V-statistic that is based on the function $\delta$.
The following result is a direct consequence of Theorem 1 of \cite{Sen:1977aa},
when applied to V-statistics constructed from functionals $\delta$
that satisfy the hypothesis of \cite[Prop. 3]{Szekely2013}.
\begin{lem}
\label{lem: Sen lemma}Make the same assumptions regarding $\mathbf{X}_{n}$
and $\mathbf{Y}_{m}$ as in Theorem \ref{thm: Asymptotic pseudo-posteriori-1}.
Let $\delta\left(\bm{x},\bm{y}\right)=\delta\left(\bm{x}-\bm{y}\right)$
be a continuous and strictly negative definite function.
If
\begin{equation}
\mathbb{E}\left(\left|\kappa_{\delta}\left(\bm{X}_{1,}\bm{X}_{2};\bm{Y}_{1},\bm{Y}_{2}\right)\right|\log^{+}\left|\kappa_{\delta}\left(\bm{X}_{1,}\bm{X}_{2};\bm{Y}_{1},\bm{Y}_{2}\right)\right|\right)<\infty\text{,}\label{eq: Sen condition}
\end{equation}
then $\mathcal{V}_{\delta}\left(\mathbf{X}_{n},\mathbf{Y}_{m}\right)$ converges
almost surely to $\mathcal{E}_{\delta}\left(\bm{X}_{1},\bm{Y}_{1}\right)\ge0$,
as $\min\left\{ n,m\right\} \rightarrow\infty$, where $\bm{X}_{1},\bm{X}_{2}\in\mathbb{X}$
and $\bm{Y}_{1},\bm{Y}_{2}\in\mathbb{X}$ are arbitrary elements of
$\mathbf{X}_{n}$ and $\mathbf{Y}_{m}$, respectively. 
\end{lem}
We may apply the result of Lemma \ref{lem: Sen lemma} directly to
the case of $\delta=\delta_{1}$ in order to provide an almost sure
convergence result regarding the V-statistic $\mathcal{V}_{\delta_{1}}\left(\mathbf{X}_{n},\mathbf{Y}_{m}\right)$.
\begin{cor}
\label{cor: Energy distance}Make the same assumptions regarding $\mathbf{X}_{n}$
and $\mathbf{Y}_{m}$ as in Theorem \ref{thm: Asymptotic pseudo-posteriori-1}.
If $\bm{X}\in\mathbb{X}$ and $\bm{Y}\in\mathbb{X}$
are arbitrary elements of $\mathbf{X}_{n}$ and $\mathbf{Y}_{m}$,
respectively, and
\begin{equation}
\mathbb{E}\left(\left\Vert \bm{X}\right\Vert _{2}^{2}\right)+\mathbb{E}\left(\left\Vert \bm{Y}\right\Vert _{2}^{2}\right)<\infty
\text{,}\label{eq: Square condition}
\end{equation}
and if $\min\left\{ n,m\right\} \rightarrow\infty$, then $\mathcal{V}_{\delta_{1}}\left(\mathbf{X}_{n},\mathbf{Y}_{m}\right)$
converges almost surely to 
\begin{equation}
\mathcal{E}\left(\bm{X},\bm{Y}\right)=\frac{\Gamma\left(\frac{d+1}{2}\right)}{\pi^{\left(d+1\right)/2}}\int_{\mathbb{R}^{d}}\frac{\left|\varphi\left(\bm{t};\bm{\theta}_{0}\right)-\varphi\left(\bm{t};\bm{\theta}\right)\right|^{2}}{\left\Vert \bm{t}\right\Vert _{2}^{d+1}}\ddr   \bm{t}\text{,}\label{eq: Specific ES}
\end{equation}
where $\varphi\left(\bm{t};\bm{\theta}\right)$ is the characteristic
function corresponding to the PDF $f\left(\bm{y};\bm{\theta}\right)$. 
\end{cor}
\begin{proof}
By the law of total expectation, we apply Lemma \ref{lem: Sen lemma}
by considering the two cases of (\ref{eq: Sen condition}): when $\left|\kappa_{\delta_{1}}\right|\le1$
and when $\left|\kappa_{\delta_{1}}\right|>1$, separately, to write
\begin{align}
&\mathbb{E}\left(\left|\kappa_{\delta_{1}}\right|\log^{+}\left|\kappa_{\delta_{1}}\right|\right)=p_{0}\mathbb{E}\left(\left|\kappa_{\delta_{1}}\right|\log^{+}\left|\kappa_{\delta_{1}}\right||\left|\kappa_{\delta_{1}}\right|\le1\right)\nonumber \\
&\quad\quad\quad+p_{1}\mathbb{E}\left(\left|\kappa_{\delta_{1}}\right|\log^{+}\left|\kappa_{\delta_{1}}\right||\left|\kappa_{\delta_{1}}\right|>1\right)\text{,}\label{eq: Law of total}
\end{align}
where $p_{0}=\mathbb{P}\left(\left|\kappa_{\delta_{1}}\right|\le1\right)$
and $p_{1}=\mathbb{P}\left(\left|\kappa_{\delta_{1}}\right|>1\right)$.
The first term on the right-hand side of (\ref{eq: Law of total})
is equal to zero, since $\log^{+}\left|\kappa_{\delta_{1}}\right|=\log\left(1\right)=0$,
whenever $\left|\kappa_{\delta_{1}}\right|\le1$. Thus, we need only
be concerned with bounding the second term.

For $\left|\kappa_{\delta_{1}}\right|>1$, $\left|\kappa_{\delta_{1}}\right|\log\left|\kappa_{\delta_{1}}\right|\le\left|\kappa_{\delta_{1}}\right|^{2}$,
thus
\[
\mathbb{E}\left(\left|\kappa_{\delta_{1}}\right|\log^{+}\left|\kappa_{\delta_{1}}\right||\left|\kappa_{\delta_{1}}\right|>1\right)\le\mathbb{E}\left(\left|\kappa_{\delta_{1}}\right|^{2}|\left|\kappa_{\delta_{1}}\right|>1\right)
\]
The condition that $\mathbb{E}\left(\left|\kappa_{\delta_{1}}\right|\log^{+}\left|\kappa_{\delta_{1}}\right|\right)<\infty$
is thus fulfilled if $\mathbb{E}\left(\left|\kappa_{\delta_{1}}\right|^{2}|\left|\kappa_{\delta_{1}}\right|>1\right)<\infty$,
which is equivalent to

\begin{align*}
\mathbb{E}\left(\left|\kappa_{\delta_{1}}\right|^{2}\right)=&p_{0}\mathbb{E}\left(\left|\kappa_{\delta_{1}}\right|^{2}|\left|\kappa_{\delta_{1}}\right|\le1\right)\\
+&p_{1}\mathbb{E}\left(\left|\kappa_{\delta_{1}}\right|^{2}|\left|\kappa_{\delta_{1}}\right|>1\right)<\infty\text{,}
\end{align*}
by virtue of the integrability of $\left\{ \left|\kappa_{\delta_{1}}\right|^{2}|\left|\kappa_{\delta_{1}}\right|\le1\right\} $
implying the existence of
$$\mathbb{E}\left(\left|\kappa_{\delta_{1}}\right|^{2}|\left|\kappa_{\delta_{1}}\right|\le1\right),$$
since it is defined on a bounded support.

Next, by the triangle inequality,
$$\left|\kappa_{\delta_{1}}\right|\le2\left(\left\Vert \bm{X}_{1}\right\Vert _{2}+\left\Vert \bm{X}_{2}\right\Vert _{2}+\left\Vert \bm{Y}_{1}\right\Vert _{2}+\left\Vert \bm{Y}_{2}\right\Vert _{2}\right)\text{,}$$
and hence
\begin{align*}
\left|\kappa_{\delta_{1}}\right|^{2} & \le4\left(\left\Vert \bm{X}_{1}\right\Vert _{2}^{2}+\left\Vert \bm{X}_{2}\right\Vert _{2}^{2}+\left\Vert \bm{Y}_{1}\right\Vert _{2}^{2}+\left\Vert \bm{Y}_{2}\right\Vert _{2}^{2}\right)\\
 & +8(\left\Vert \bm{X}_{1}\right\Vert _{2}\left\Vert \bm{X}_{2}\right\Vert _{2}+\left\Vert \bm{X}_{1}\right\Vert _{2}\left\Vert \bm{Y}_{1}\right\Vert _{2}+\left\Vert \bm{X}_{1}\right\Vert _{2}\left\Vert \bm{Y}_{2}\right\Vert _{2}\\
 & +\left\Vert \bm{X}_{2}\right\Vert _{2}\left\Vert \bm{Y}_{1}\right\Vert _{2}+\left\Vert \bm{X}_{2}\right\Vert _{2}\left\Vert \bm{Y}_{2}\right\Vert _{2}+\left\Vert \bm{Y}_{1}\right\Vert _{2}\left\Vert \bm{Y}_{2}\right\Vert _{2})\text{.}
\end{align*}
\normalsize
Since $\bm{X}_{1},\bm{X}_{2},\bm{Y}_{1},\bm{Y}_{2}$ are all pairwise
independent, and $\bm{X}_{1}$ and $\bm{Y}_{1}$ are identically distributed
to $\bm{X}_{2}$ and $\bm{Y}_{2}$, respectively, we have 
\begin{align*}
\mathbb{E}\left(\left|\kappa_{\delta_{1}}\right|^{2}\right) & 
\le8\left[\mathbb{E}\left(\left\Vert \bm{X}_{1}\right\Vert _{2}^{2}\right)+\mathbb{E}\left(\left\Vert \bm{Y}_{1}\right\Vert _{2}^{2}\right)\right]\\
&+8\left[\left(\mathbb{E}\left\Vert \bm{X}_{1}\right\Vert _{2}\right)^{2}+\left(\mathbb{E}\left\Vert \bm{Y}_{1}\right\Vert _{2}\right)^{2}\right]\\ & +32\left[\mathbb{E}\left\Vert \bm{X}_{1}\right\Vert _{2}\mathbb{E}\left\Vert \bm{Y}_{1}\right\Vert _{2}\right]\text{,}
\end{align*}
which concludes the proof since
$\mathbb{E}\left\Vert \bm{X}_{1}\right\Vert^2 _{2}+\mathbb{E}\left\Vert \bm{Y}_{1}\right\Vert^2 _{2}<\infty$ is satisfied by the hypothesis and implies $\mathbb{E}\left\Vert \bm{X}_{1}\right\Vert _{2}+\mathbb{E}\left\Vert \bm{Y}_{1}\right\Vert _{2}<\infty$.
\end{proof}
We note that condition (\ref{eq: Square condition}) is stronger than
a direct application of condition (\ref{eq: Sen condition}), which
may be preferable in some situations. However, condition (\ref{eq: Square condition})
is somewhat more intuitive and verifiable since it is concerned with
the polynomial moments of norms and does not involve the piecewise
function $\log^{+}x$. It is also suggested in \cite{Zygmund:1951aa}
that one may replace $\log^{+}x$ by $\log\left(2+x\right)$ if it
is more convenient to do so. 
We further note that (\ref{eq: Square condition}) is required for establishing almost sure convergence, and is stronger than what is needed to ensure convergence in probability, as is established in \cite{Szekely2004} and \cite{Gretton:2012aa}. 

Combining the result of Theorem \ref{thm: Asymptotic pseudo-posteriori-1}
with Corollary \ref{cor: Energy distance} and the conclusion from
Proposition 1 of \cite{Szekely2013} provided in Equation~\eqref{eq: Specific ES} yields the key result below.
This result justifies the use of the V-statistic estimator $\mathcal{V}_{\delta_{1}}\left(\mathbf{X}_{n},\mathbf{Y}_{m}\right)$
for the energy distance $\mathcal{E}\left(\bm{X},\bm{Y}\right)$ within the
IS-ABC framework, \textcolor{black}{and is comparable to Corollaries 1--3 of \cite{Jiang2018} regarding the large sample asymptotics of other discrepancy measurements.}
\begin{cor}
\label{cor: Specialization to ES} Under the assumptions of Corollary
\ref{cor: Energy distance}. If $\mathcal{D}\left(\mathbf{X}_{n},\mathbf{Y}_{m}\right)=\mathcal{V}_{\delta_{1}}\left(\mathbf{X}_{n},\mathbf{Y}_{m}\right)$,
then the conclusion of Theorem \ref{thm: Asymptotic pseudo-posteriori-1}
follows with
\begin{align*}
\mathcal{D}\left(\mathbf{X}_{n},\mathbf{Y}_{m}\right)\rightarrow&\frac{\Gamma\left(\frac{d+1}{2}\right)}{\pi^{\left(d+1\right)/2}}\int_{\mathbb{R}^{d}}\frac{\left|\varphi\left(\bm{t};\bm{\theta}_{0}\right)-\varphi\left(\bm{t}; \bm{\theta}\right)\right|^{2}}{\left\Vert \bm{t}\right\Vert _{2}^{d+1}}\ddr  \bm{t}\\
=&\mathcal{D}_{\infty}\left(\bm{\theta}_{0},\bm{\theta}\right)\text{,}
\end{align*}
almost surely, as $n\rightarrow\infty$, where $\mathcal{D}_{\infty}\left(\bm{\theta}_{0},\bm{\theta}\right)\ge0$
and $\mathcal{D}_{\infty}\left(\bm{\theta}_{0},\bm{\theta}\right)=0$, if and
only if $\bm{\theta}_{0}=\bm{\theta}$.
\end{cor}

\subsection{Behavior as $\mathbf{\epsilon} \to \mathbf{0}$}

Let $\mathbb{F}$ be the set of probability distributions on $\mathbb{X}$.
From \cite[Thm. 22]{Sejdinovic:2013aa}, we have the fact that $\mathcal{E}^{1/2}\left(\bm{X},\bm{Y}\right)=\mathcal{E}^{1/2}\left(F_{X},F_{Y}\right)$
is a metric on $\mathbb{F}$, where $\bm{X}$ and $\bm{Y}$ have data
generating process that are characterized by $F_{X}$ and $F_{Y}$,
respectively. As such, when we
take $\bm{X}$ and $\bm{Y}$ arising from two empirical distributions
with an equal number of masses (defined on $\mathbf{x}_{n}$ and $\mathbf{y}_{n}$,
for instance), then we obtain the fact that $\mathcal{E}^{1/2}\left(\bm{X},\bm{Y}\right)=0$
if and only if the two empirical distributions are the same. In other words,
$\mathbf{x}_{n}$ and $\mathbf{y}_{n}$ are equal, in the sense that
the elements of $\mathbf{x}_{n}$ and $\mathbf{y}_{n}$ are equal
up to a permutation. Proposition 2 of \cite{Bernton:2017aa} then provides the following
result in the case when $n$ is fixed.
\begin{prop} \label{prop: Bernton Prop 2}
Assume that $w\left(d,\epsilon\right)$ has form \eqref{eq: rejection}, where $d=\mathcal{V}_{\delta_{1}}$,
and that $f\left(\mathbf{x}_{n}|\bm{\theta}\right)$ is a continuous
and exchangeable PDF. Furthermore, assume that
\[
\sup_{\bm{\theta}\in\mathbb{T}\backslash\left\{ \Theta\subset\mathbb{T}:\int_{\Theta}\pi\left(\bm{\theta}\right)\mathrm{d}\bm{\theta}=0\right\} }f\left(\mathbf{x}_{n}|\bm{\theta}\right)<\infty
\]
and suppose that there exists some $\bar{\epsilon}>0$, where
\[
\sup_{\bm{\theta}\in\mathbb{T}\backslash\left\{ \Theta\subset\mathbb{T}:\int_{\Theta}\pi\left(\bm{\theta}\right)\mathrm{d}\bm{\theta}=0\right\} }\sup_{\left\{ \mathbf{y}_{n}:d\left(\mathbf{x}_{n},\mathbf{y}_{n}\right)\le\bar{\epsilon}\right\} }f\left(\mathbf{y}_{n}|\bm{\theta}\right)<\infty\text{.}
\]
Then, for fixed $\mathbf{x}_{n}$, the pseudo-posterior PDF \eqref{eq: pseudo posterior} converges
strongly to the posterior PDF \eqref{eq: posterior}, as $\epsilon\rightarrow 0$.
\end{prop}

Let us suppose that the empirical distribution of $\bm{X}_{n}$ is
denoted $\hat{F}_{n}$ and that each observation of $\bm{X}_{n}$
is generated from a process that can be characterized by the distribution
$F_{0}$ (we write the joint distribution of $\bm{X}_{n}$ as $F_{n}$).
We shall also write $F_{n}^{\bm{\theta}}$ as the probability distribution
corresponding to the PDF $f\left(\mathbf{x}_{n}|\bm{\theta}\right)$,
and $\hat{F}_{n}^{\bm{\theta}}$ as the empirical distribution obtained
from a sample $\bm{Y}_{n}$ with data generating process that is characterized
by $F_{n}^{\bm{\theta}}$.

Next, we let the probability distribution corresponding to the prior
and pseudo-posterior PDFs  of the ES-based ABC process
with rejection weights (i.e. $\pi(\bm{\theta})$ and \eqref{eq: pseudo posterior}) as $\Pi$ and $\Pi_{n}^{\epsilon}$, respectively.
And finally, let us denote the probability of the set $\mathbb{A}$
with respect to the probability distribution $F$ as $F\left(\mathbb{A}\right)$.
In order to state our next result, we require the following assumptions.
\begin{description}
\item [{A1}] The data generating process of $\bm{X}_{n}$ is such that,
for every $\varepsilon>0$,
\[
\lim_{n\rightarrow\infty}F_{n}\left(\mathcal{E}^{1/2}\left(\hat{F}_{n},F_{0}\right)>\varepsilon\right)=0\text{.}
\]
\item [{A2}] For every $\epsilon>0$,
\[
F_{n}^{\theta}\left(\mathcal{E}^{1/2}\left(\hat{F}_{n}^{\bm{\theta}},F_{1}^{\bm{\theta}}\right)>\epsilon\right)\le c\left(\bm{\theta}\right)s_{n}\left(\epsilon\right)
\]
where $s_{n}\left(\epsilon\right)$ is a sequence of functions that
is strictly decreasing in $\epsilon$ for all $n$, and $s_{n}\left(\epsilon\right)\rightarrow0$
as $n\rightarrow\infty$, for fixed $\epsilon$. Here: $c\left(\bm{\theta}\right)$
is a positive function that is integrable with respect to $\Pi$ and
satisfies $c\left(\bm{\theta}\right)\le c_{0}$ for some $c_{0}$,
for all $\bm{\theta}$ such that, for some $\delta_{0}>0$, $\mathcal{E}^{1/2}\left(F_{0},F_{1}^{\bm{\theta}}\right)\le\delta_{0}+\epsilon_{0}$,
where $\epsilon_{0}=\min_{\bm{\theta}\in\mathbb{T}}\mathcal{E}^{1/2}\left(F_{0},F_{1}^{\bm{\theta}}\right)$.
\item [{A3}] There exists an $L>0$ and a $c_{\pi}>0$ such that, for each
 sufficiently small $\epsilon>0$,
\[
\Pi\left(\bm{\theta}\in\mathbb{T}:\mathcal{E}^{1/2}\left(F_{0},F_{1}^{\bm{\theta}}\right)\le\epsilon+\epsilon_{0}\right)\ge c_{\pi}\epsilon^{L}\text{.}
\]
\end{description}
Upon making Assumptions A1--A3, we may apply the proof process of
\cite[Prop. 3]{Bernton:2017aa} directly, replacing the Wasserstein metric with the
energy metric $\mathcal{E}^{1/2}$, where appropriate. Such a process
yields the following result.
\begin{prop} \label{prop: Bernton Prop 3}
Along with A1--A3, assume that there
exists a sequence $\left\{ \epsilon_{n}\right\} _{n=1}^{\infty}$,
such that $\epsilon_{n}\rightarrow0$, $s_{n}\left(\epsilon_{n}\right)\rightarrow0$
and $F_{n}\left(\mathcal{E}^{1/2}\left(\hat{F}_{n},F_{0}\right)\le\epsilon_{n}\right)\rightarrow1$,
as $n\rightarrow\infty$. Then, the ES-based ABC algorithm with $m=n$,
discrepancy $d=\mathcal{V}_{\delta_{1}}^{1/2}$, and rejection weights
using $\epsilon=\epsilon_{n}+\epsilon_{0}$ satisfies the inequality
\[
\Pi_{n}^{\epsilon_{n}+\epsilon_{0}}\left[\mathcal{E}^{1/2}\left(F_0,F_{1}^{\bm{\theta}}\right)>\epsilon_{0}+\frac{4\epsilon_{n}}{3}+s_{n}^{-1}\left(\frac{\epsilon_{n}^{L}}{R}\right)\right]\le\frac{C}{R}\text{,}
\]
for some $C,R\in\left(0,\infty\right)$, with probability going to
1 as $n\rightarrow\infty$ (with respect to $F_0$).
\end{prop}
The hypotheses of Proposition \ref{prop: Bernton Prop 2} are straight
forward and the conclusion implies that pseudo-posterior PDF of the
ES-based ABC procedure can approximate the posterior PDF, based on
the likelihood of the data generating process of $\mathbf{x}_{n}$,
to an arbitrary level of accuracy, when $\epsilon$ is made sufficiently
small. This however does not mean that one should make $\epsilon$ too
small in practice, as the effort required to simulate data will become
more difficult and the process becomes more computationally intensive
in such cases. \textcolor{black}{Note that the value of $\epsilon$ is often chosen in a pragmatic way as a quantile (of a small order, usually less than 5\%) of all the distances that are obtained in the ABC sample, thus deciding how many samples are kept as a fraction of the entire ABC replications. This procedure was used in the ABC algorithms of \cite{beaumont2002approximate}, \cite{blum2010choosing}, and \cite{jabot2013easy}.
}

Next, we note that the assumptions (other than A1, which is generally true for stationary and ergodic data; cf. \cite[Secs. 7 and 8]{Szekely2013}) and the conclusion
of Proposition \ref{prop: Bernton Prop 3} are more complex. Due to
the lack of ease by which A2 and A3 may be validated, the proposition
is more useful as an existence result regarding what can be expected
in theory, with respect to how quickly the ES-based ABC algorithm
converges in $n$, rather than providing any practical guidance. \textcolor{black}{A suggestion by \cite{Bernton:2017aa} is that one may potentially apply the theory of \cite{fournier2015rate} and \cite{weed2019sharp} in order to validate assumption A2.}

Under further assumptions, the concentration with respect
to the discrepancy in distributions can be transferred to a concentration
result, with respect to parameter vector in the space $\mathbb{T}$ (cf. \cite[Cor. 1]{Bernton:2017aa}).

\subsection{Illustration on a simple example}\label{sec:theory-illustration}

We use $\bm{X}\sim\mathcal{L}$ to denote that the random variable $\bm{X}$ has probability law $\mathcal{L}$. Furthermore, we denote the normal law by $\mathcal{N}\left(\bm{\mu},\bm{\Sigma}\right)$, where $\bm{X}\sim\mathcal{N}\left(\bm{\mu},\bm{\Sigma}\right)$ states that the DGP of $\bm{X}$ is multivariate normal distribution with mean vector $\bm{\mu}$ and covariance matrix $\bm{\Sigma}$.

For illustrating the theoretical results, we investigate the pseudo-posterior limit on a simple univariate Gaussian location model  $\mathcal{N}(\bm{\theta},\sigma^2)$ (with known variance $\sigma^2$) with conjugate Gaussian prior $\bm{\theta}\sim \mathcal{N}(0,\tau^2)$ (with variance $\tau^2$ fixed). We have  IID observations $\bm{X}_{1},\dots,\bm{X}_{n}\mid \bm{\theta}_0 \sim \mathcal{N}(\bm{\theta}_0,\sigma^2)$, and  IID replicates $\bm{Y}_{1},\dots,\bm{Y}_{m}\mid \bm{\theta} \sim \mathcal{N}(\bm{\theta},\sigma^2)$. The posterior is 
$\bm{\theta}\mid \bm{X}_{1},\dots,\bm{X}_{n} \sim \mathcal{N}(\hat{\bm{\theta}},\hat{\sigma}^2)$, where 
\begin{equation*}
    \hat{\bm{\theta}} = \frac{n\bar{\bm{X}}_n}{n+\sigma^2/\tau^2}, 
    \quad
    \hat{\sigma}^{-2} = \tau^{-2}+n\sigma^{-2}.
\end{equation*}

In this simple model, the limiting data discrepancy takes the form (up to a proportionality constant) of $\mathcal{D}_{\infty}\left(\bm{\theta}_{0},\bm{\theta}\right) = (\bm{\theta}_{0}-\bm{\theta})^2$ for the energy distance and Kullback--Leibler divergence, and $\mathcal{D}_{\infty}\left(\bm{\theta}_{0},\bm{\theta}\right) = |\bm{\theta}_{0}-\bm{\theta}|$ for the MMD and the (second order) Wasserstein distance.

Theorem~\ref{thm: Asymptotic pseudo-posteriori-1} establishes that the large $n$ and $m$ limit of the pseudo-posterior $\pi_{m,\epsilon}$ is the distribution that we denote here by $\pi_{\infty,\epsilon}\left(\bm{\theta}\right)\propto \pi\left(\bm{\theta}\right)w\left(\mathcal{D}_{\infty}\left(\bm{\theta}_{0},\bm{\theta}\right),\epsilon\right)$. 
For illustrative purposes, let us focus on the case when $\mathcal{D}_{\infty}\left(\bm{\theta}_{0},\bm{\theta}\right) = \vert\bm{\theta}_{0}-\bm{\theta}\mid$, and consider rejection ABC with 
$w\left(d,\epsilon\right)=\left\llbracket d<\epsilon\right\rrbracket$
 and IS-ABC with
$w\left(d,\epsilon\right)=\exp\left(-d^2/2\epsilon^2\right)$. The limiting pseudo-posterior can then be obtained in closed-form as
\begin{align}
    \pi_{\infty,\epsilon}\left(\bm{\theta}\right) 
    &\propto \mathcal{N}(\bm{\theta}\mid 0, \tau^2) \left\llbracket \vert\bm{\theta}-\bm{\theta}_0\vert <\epsilon\right\rrbracket,\nonumber\\
    & = \mathcal{N}_{[\bm{\theta}_0-\epsilon, \bm{\theta}_0+\epsilon]}(\bm{\theta}\mid 0, \tau^2), \label{eq: pseudo posterior reject}
\end{align}
a truncated Gaussian for rejection ABC and
\begin{align}
    \pi_{\infty,\epsilon}\left(\bm{\theta}\right) 
    &\propto \mathcal{N}(\bm{\theta}\mid 0, \tau^2) \exp\left(\frac{(\bm{\theta}_{0}-\bm{\theta})^2}{2\epsilon^2}\right),\nonumber\\
    & = \mathcal{N}(\bm{\theta}\mid \bar{\bm{\theta}}(\epsilon), \bar{\sigma}^2(\epsilon)), \label{eq: pseudo posterior IS}
\end{align}
for IS-ABC, where $\bar{\bm{\theta}}(\epsilon)=\frac{\bm{\theta}_0}{1+\epsilon^2/\tau^2}$ and $\bar{\sigma}^{-2}(\epsilon) = \tau^{-2}+\epsilon^{-2}$. See Figure~\ref{fig: illustration-theory} for an illustration, for various values of $\epsilon$.

\begin{figure}[t]
    \centering
    \includegraphics[width=8cm]{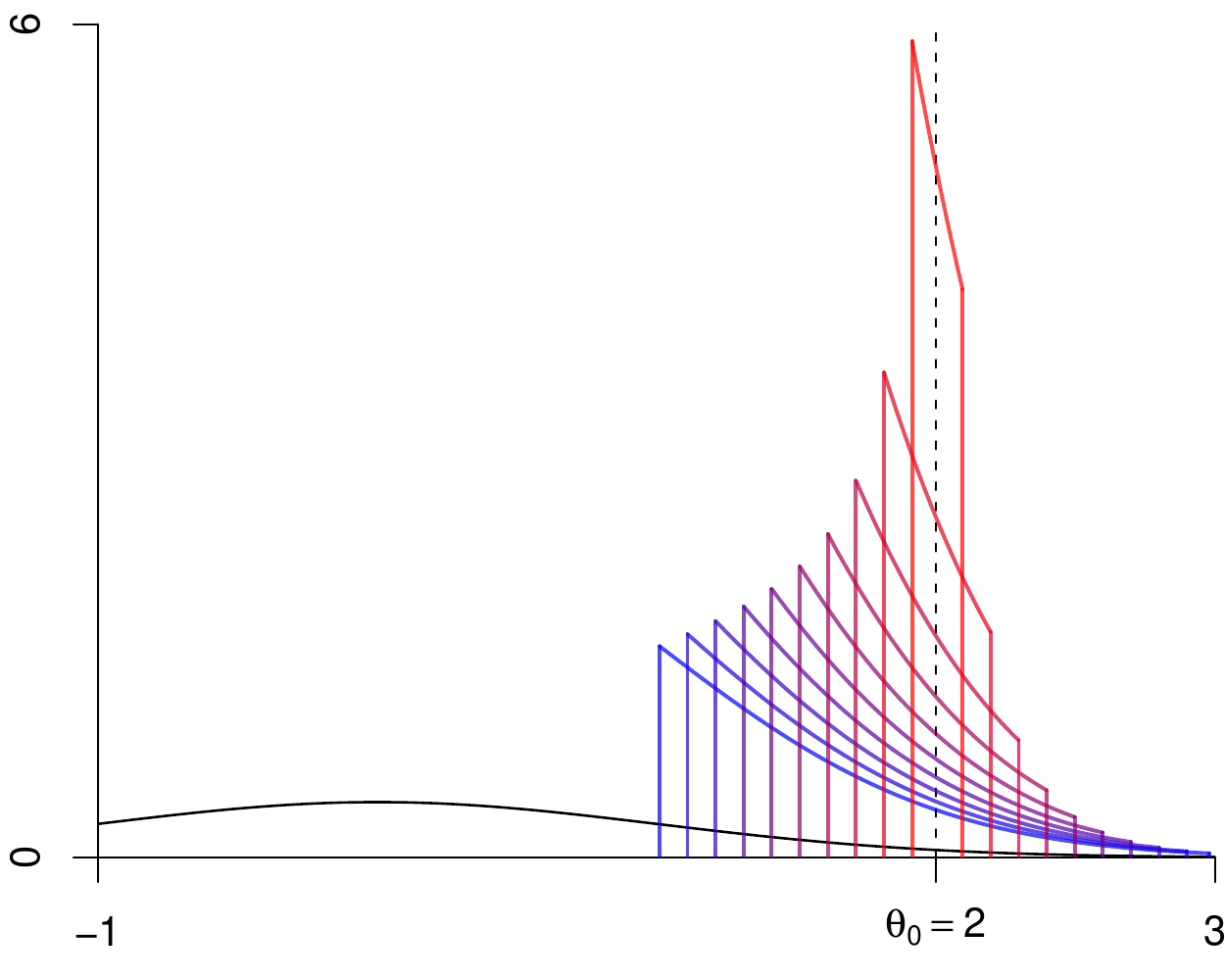}
    \includegraphics[width=8cm]{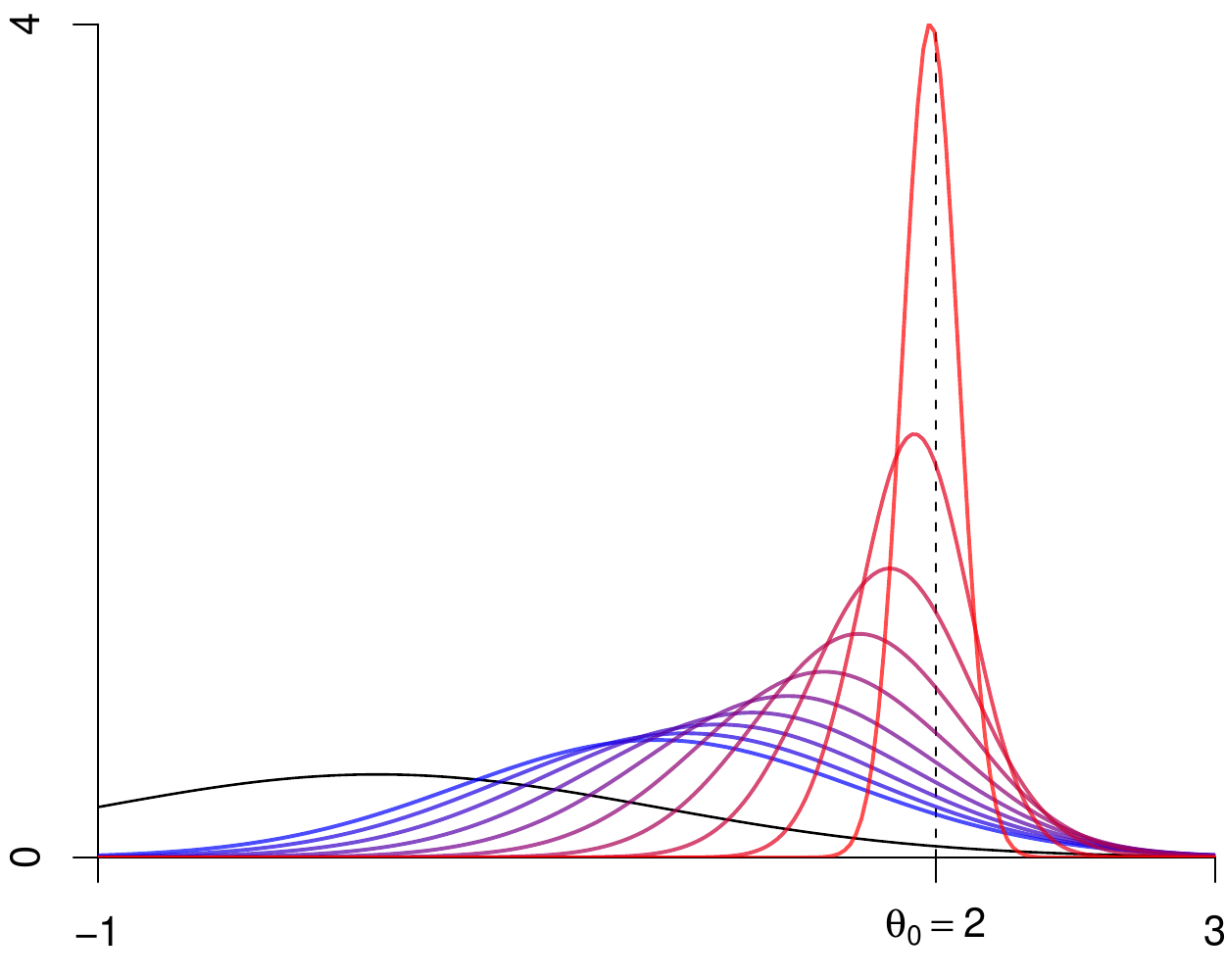}
    \caption{Illustration of theoretical results under rejection ABC (left) and IS-ABC (right) for the univariate Gaussian location model (see details in Section~\ref{sec:theory-illustration}. The prior density is displayed in black, and the true parameter $\bm{\theta}_0=2$ is indicated by a vertical dashed line. Limiting pseudo-posterior densities $\pi_{\infty,\epsilon}$ provided in \eqref{eq: pseudo posterior reject} (left) and \eqref{eq: pseudo posterior IS} (right), are computed with $\epsilon$ varying in $\{1, 0.9, \ldots,0.2, 0.1\}$ with colors from blue ($\epsilon = 1$) to red  ($\epsilon = 0.1$).}
    \label{fig: illustration-theory}
\end{figure}

\section{Illustrations} \label{sec:illustrations}

We illustrate the use of the ES on some standard models. The standard rejection ABC algorithm is employed (that is, we use Algorithm~\ref{alg:IS-ABC} with weight function $w$ of form~\eqref{eq: rejection}) for constructing estimators~\eqref{eq: Sample approx Bayes est}. The proposed ES is compared to the Kullback--Leibler divergence (KL), the Wasserstein distance (WA), and the maximum mean discrepancy (MMD). Here, the ES is applied using the Euclidean metric $\delta_1$, the Wasserstein distance using the exponent $p=2$ and \textcolor{black}{the approximation by the swapping distance} \citep{Bernton:2017aa} and the MMD using a Gaussian kernel $\chi(\bm{x},\bm{y})=\exp{[-(\bm{x}-\bm{y})^2]}$. The Gaussian kernel is commonly used in the MMD literature, and was also considered for ABC in \cite{Park2016} and \cite{Jiang2018}. Details regarding the use of the Kullback--Leibler divergence as a discrepancy function for ABC algorithms can be found in Sec. 2 of \cite{Jiang2018}. \textcolor{black}{With respect to the theoretical results of Section~\ref{sec:theory}, the chosen examples can be shown to be sufficiently regular as to validate the hypotheses of Corollary~\ref{cor: Specialization to ES} and Proposition~\ref{prop: Bernton Prop 2}. However, we believe that it would be difficult to validate Assumptions A2 and A3 of Proposition~\ref{prop: Bernton Prop 3}, without further theoretical development.}

We consider examples explored in \cite[Sec. 4.1]{Jiang2018}. For each illustration below, we sample synthetic data of the same size $m$ as the observed data size, $n$, whose value is specified for each model below. \textcolor{black}{The ABC procedure is sensitive to the choice of the prior; we follow the benchmark examples of \cite{Jiang2018} by employing  the same uniform priors, as specified in each example.} 
The number of ABC iterations in Algorithm~\ref{alg:IS-ABC} is set to $N=10^5$. The tuning parameter $\epsilon$ is set so that only the $0.05\%$ smallest discrepancies are kept to form ABC posterior sample. We postpone the discussion of the results of our simulation experiments to Section~\ref{sec:illustration-discussion}

The experiments were implemented in \textsf{R}, using in particular the \textsf{winference} package \citep{Bernton:2017aa} and the \textsf{FNN} package \citep{FNN2013}. 
The Kullback--Leibler divergence between two PDFs is computed within 
the $1$-nearest neighbor framework \citep{Boltz2009}. Moreover, 
the $k$-d trees is adopted for implementing the nearest neighbor search, 
which is the same as the method of \cite{Jiang2018}. For estimating 
the $2$-Wasserstein distance between two multivariate empirical measures, we propose 
to employ the swapping algorithm \citep{PUCCETTI2017132}, which is simple 
to implement, and is more accurate and less computationally expensive than other 
algorithms commonly used in the literature \citep{Bernton:2017aa}. Regarding the MMD, the same unbiased U-statistic estimator 
is adopted as given in \cite{Jiang2018} and \cite{Park2016}. For reproduction of the 
the experimental results, the original source code can be accessed at 
\url{https://github.com/hiendn/Energy_Statistics_ABC}.

\subsection{Bivariate Gaussian mixture model} \label{sec:BGM}

Let $\mathbf{X}_n$ be a sequence of IID random variables, such that each $\bm{X}_i$ has a mixture of bivariate Gaussian probability law
\begin{equation}\label{eq: Gaussian_mixture}
    \bm{X}_i\sim p\mathcal{N}(\bm{\mu}_0,\bm{\Sigma}_0) + (1-p)\mathcal{N}(\bm{\mu}_1,\bm{\Sigma}_1),
\end{equation}
with known covariance matrices 
\[
\bm{\Sigma}_{0}=\left[\begin{array}{cc}
0.5 & -0.3\\
-0.3 & 0.5
\end{array}\right]\text{ and }\bm{\Sigma}_{1}=\left[\begin{array}{cc}
0.25 & 0\\
0 & 0.25
\end{array}\right]\text{.}
\]
We aim to estimate the generative parameters $\bm{\theta}^\top=(p,\bm{\mu}_0^\top,\bm{\mu}_1^\top)$ consisting of the mixing probability $p$ and the population means $\bm{\mu}_0$ and $\bm{\mu}_1$. 
We denote the uniform law, in the interval $(a,b)$, for $a<b$, by $\text{Unif}(a,b)$. The priors  on the model parameters are uniform; 
that is, $\bm{\mu}_1 \sim \text{Unif}(-1,1)^2$, $\bm{\mu}_2 \sim \text{Unif}(-1,1)^2$ and  $p \sim \text{Unif}(0,1)$.  
We perform ABC using $n = 500$ observations, sampled from model~\eqref{eq: Gaussian_mixture} with $p=0.3$, $\bm{\mu}_0^\top=(0.7,0.7)$ and $\bm{\mu}_1^\top=(-0.7,-0.7)$.  A kernel density estimate (KDE) of the ABC posterior distribution (bivariate marginals of $\bm{\mu}_0$ and $\bm{\mu}_1$) is presented in Figure~\ref{fig: Gaussian mixture}.

\begin{figure*}[t]
    \centering
    \includegraphics[width=5cm]{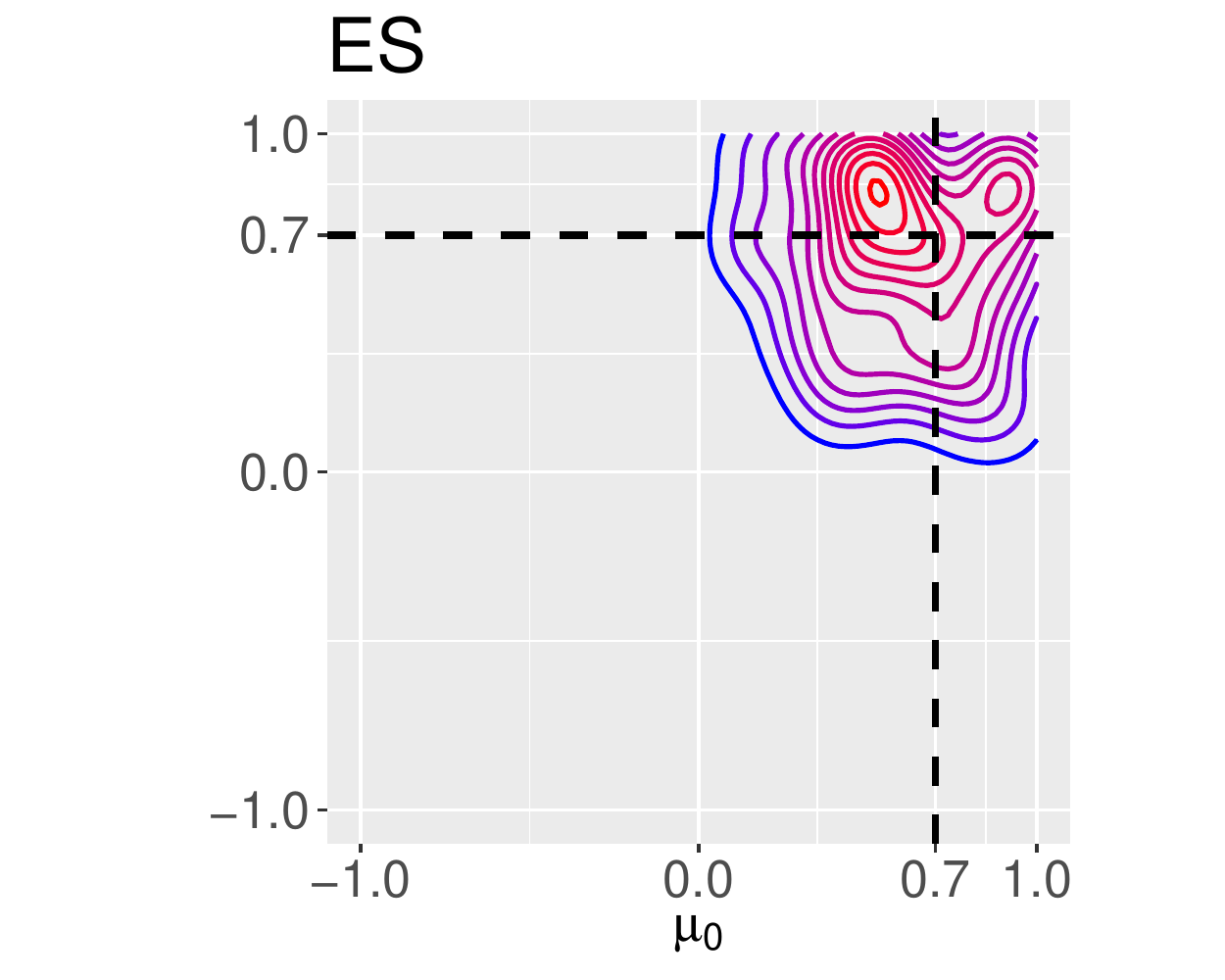}
    \hspace{-1.4cm}
    \includegraphics[width=5cm]{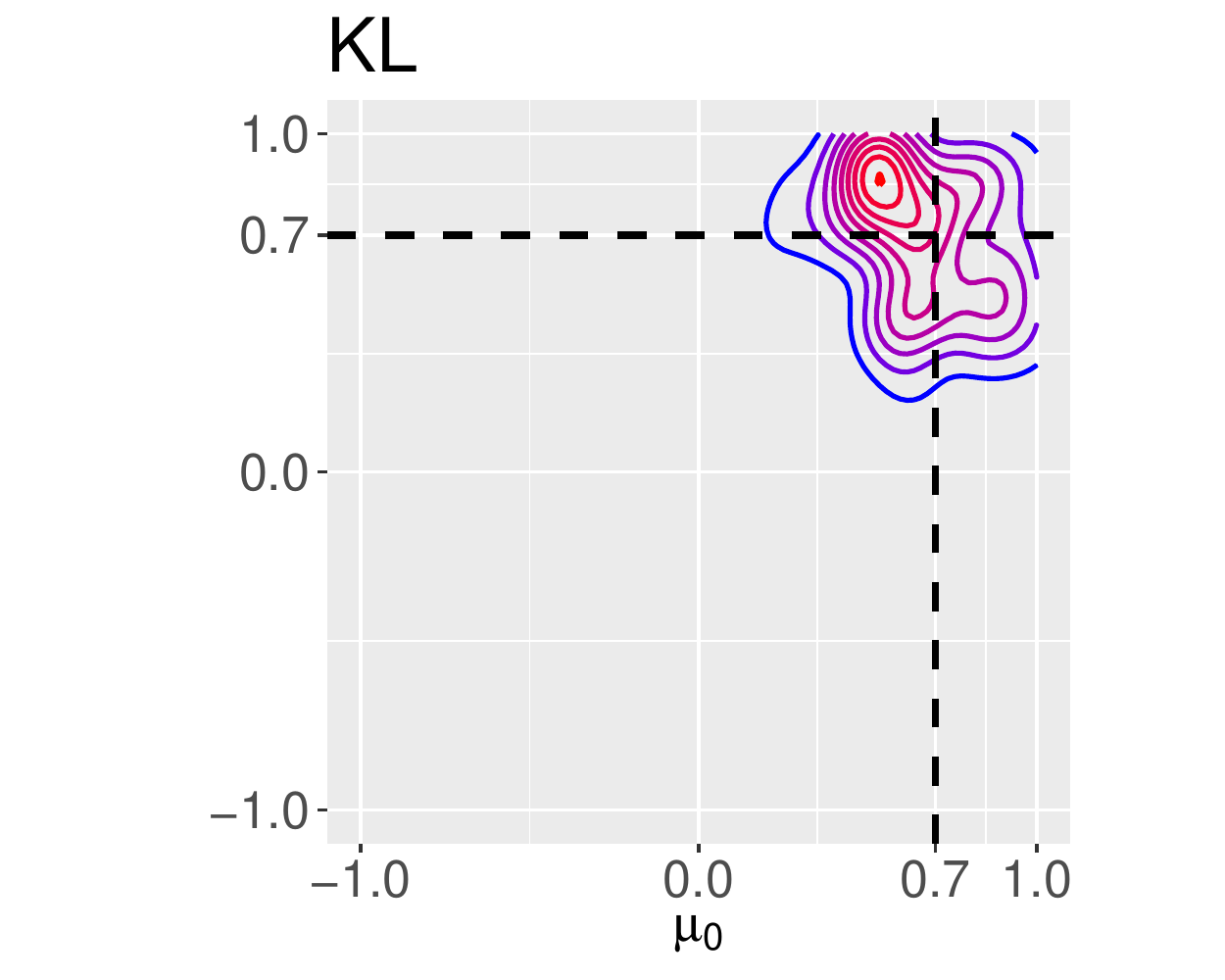}
    \hspace{-1.4cm}
    \includegraphics[width=5cm]{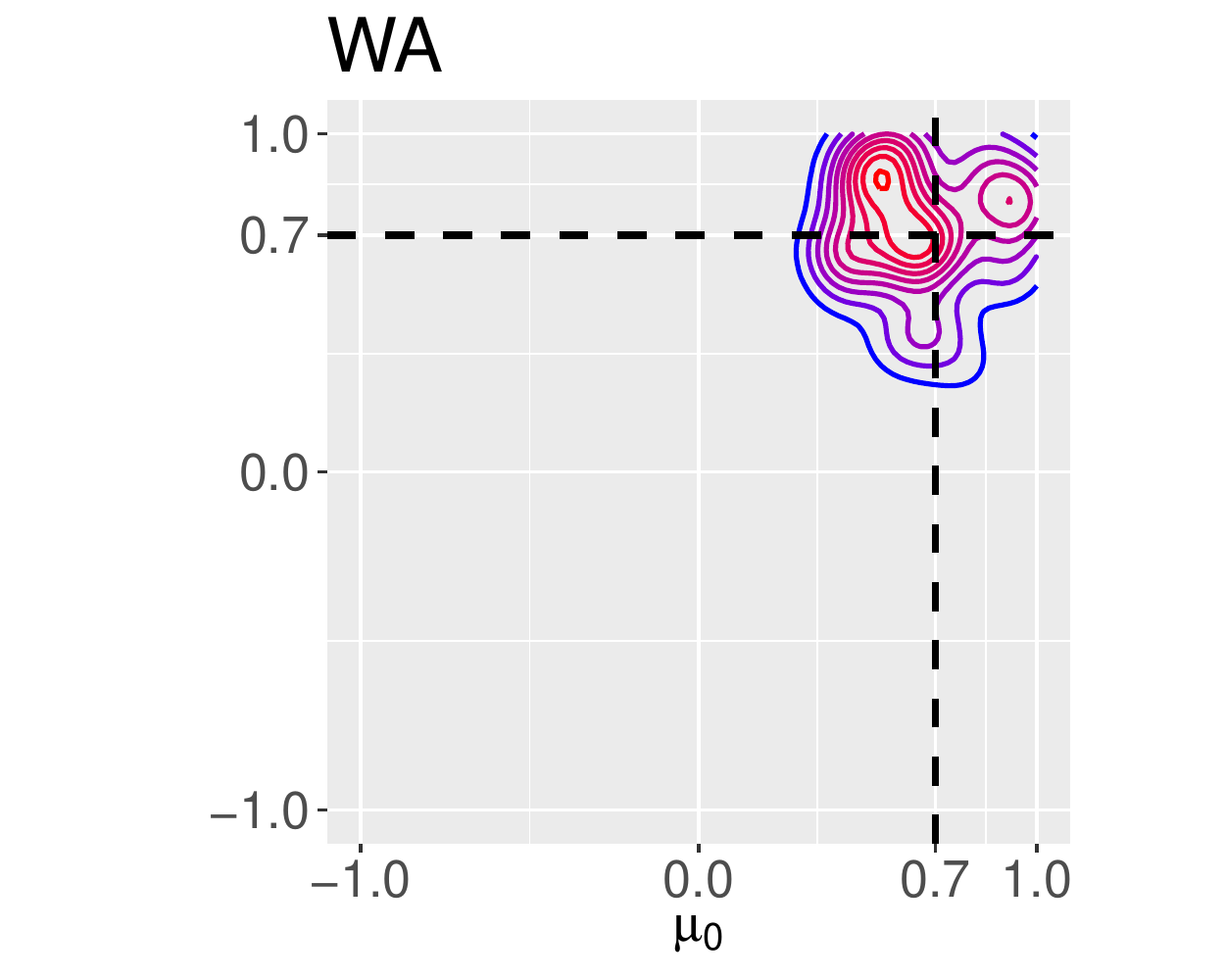}
    \hspace{-1.4cm}
    \includegraphics[width=5cm]{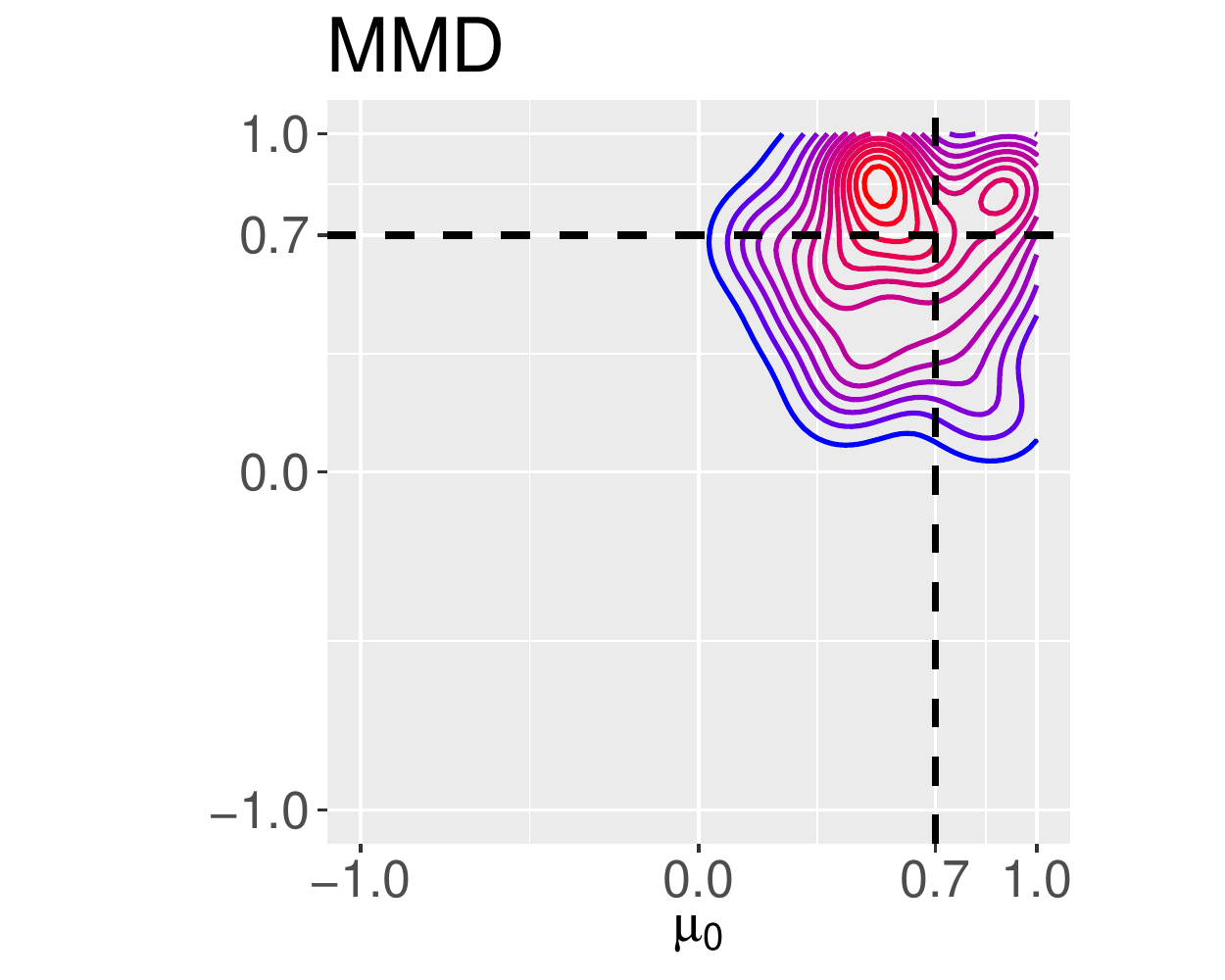}
    \hspace{-1.4cm}
    \includegraphics[width=5cm]{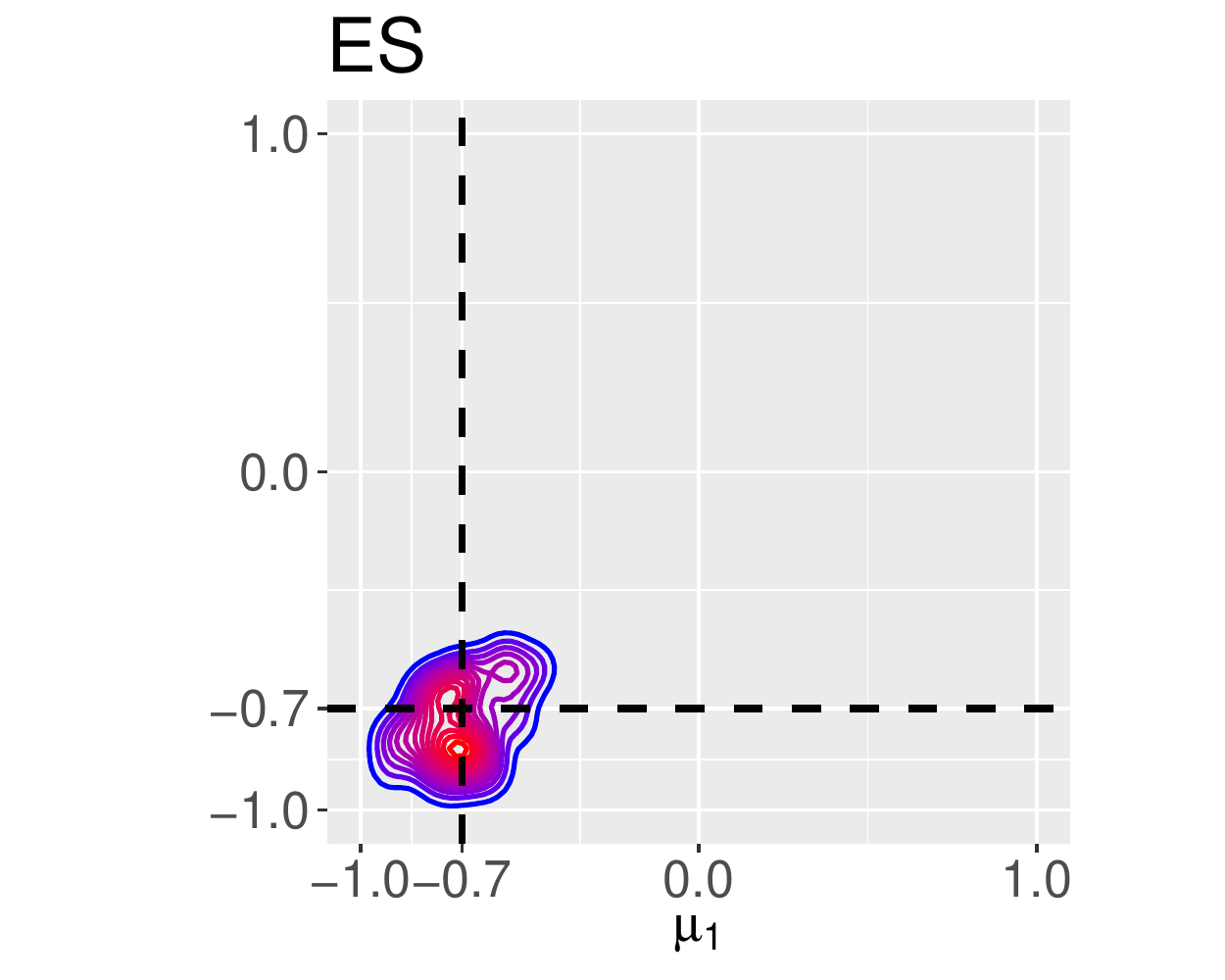}
    \hspace{-1.4cm}
    \includegraphics[width=5cm]{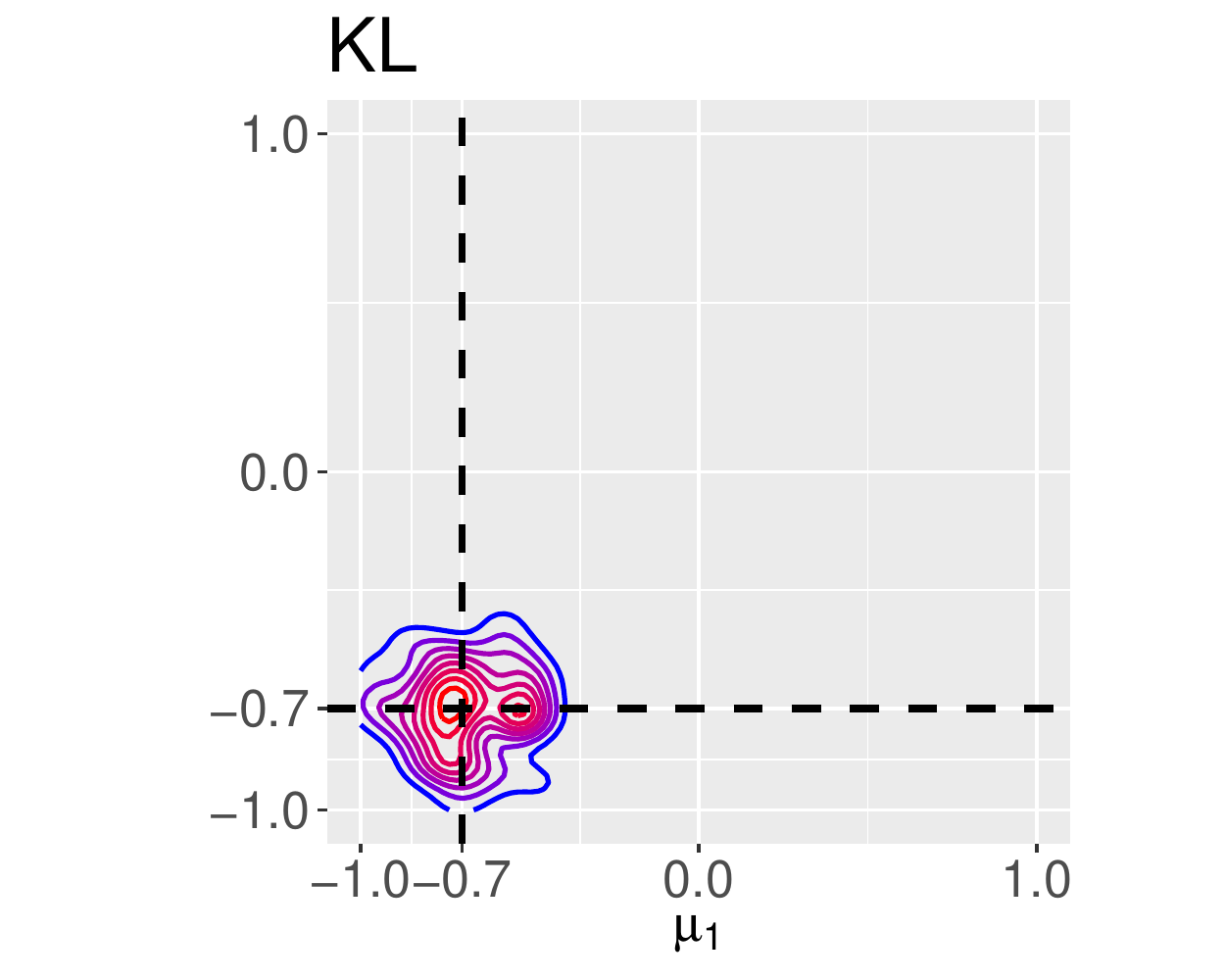}
    \hspace{-1.4cm}
    \includegraphics[width=5cm]{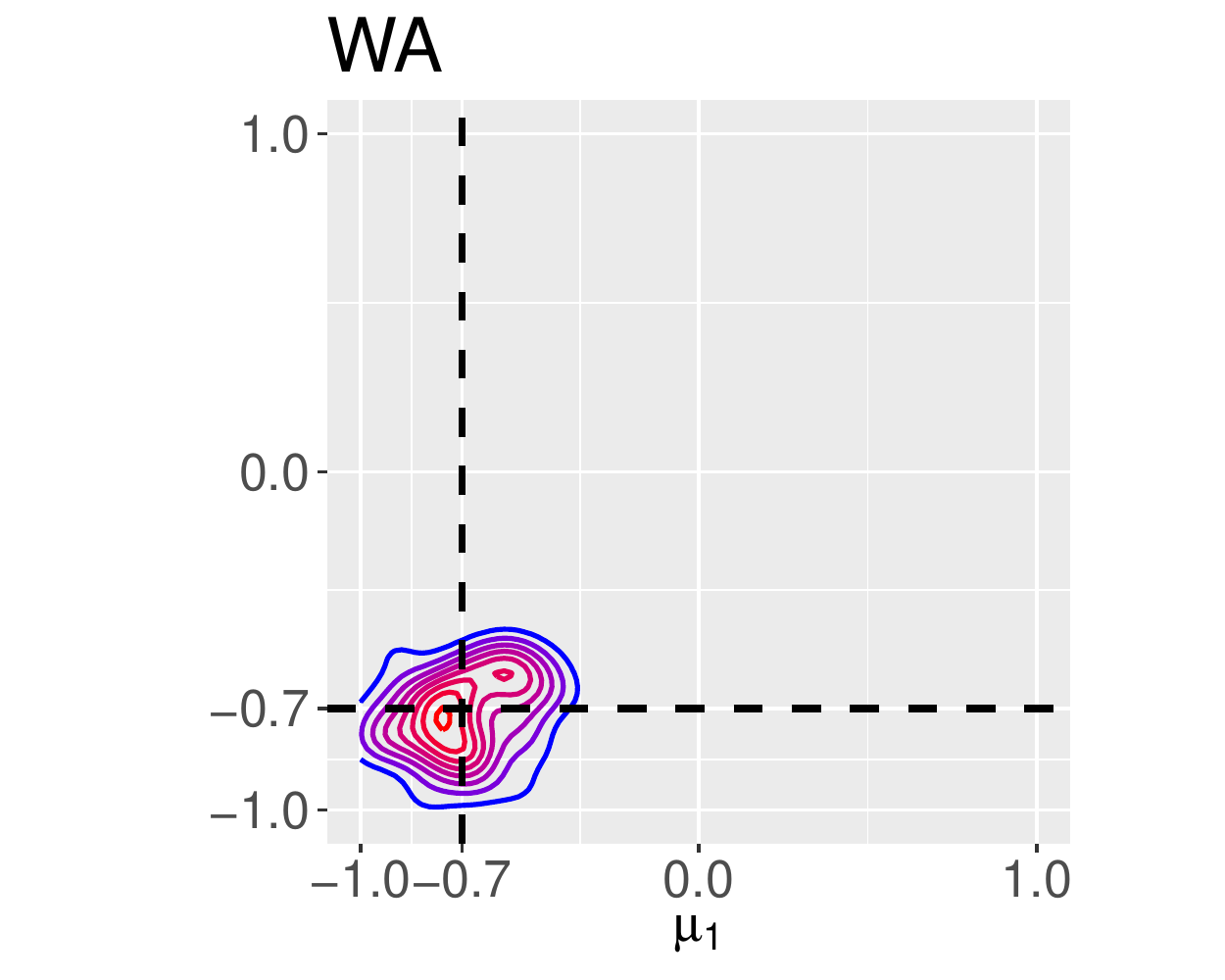}
    \hspace{-1.4cm}
    \includegraphics[width=5cm]{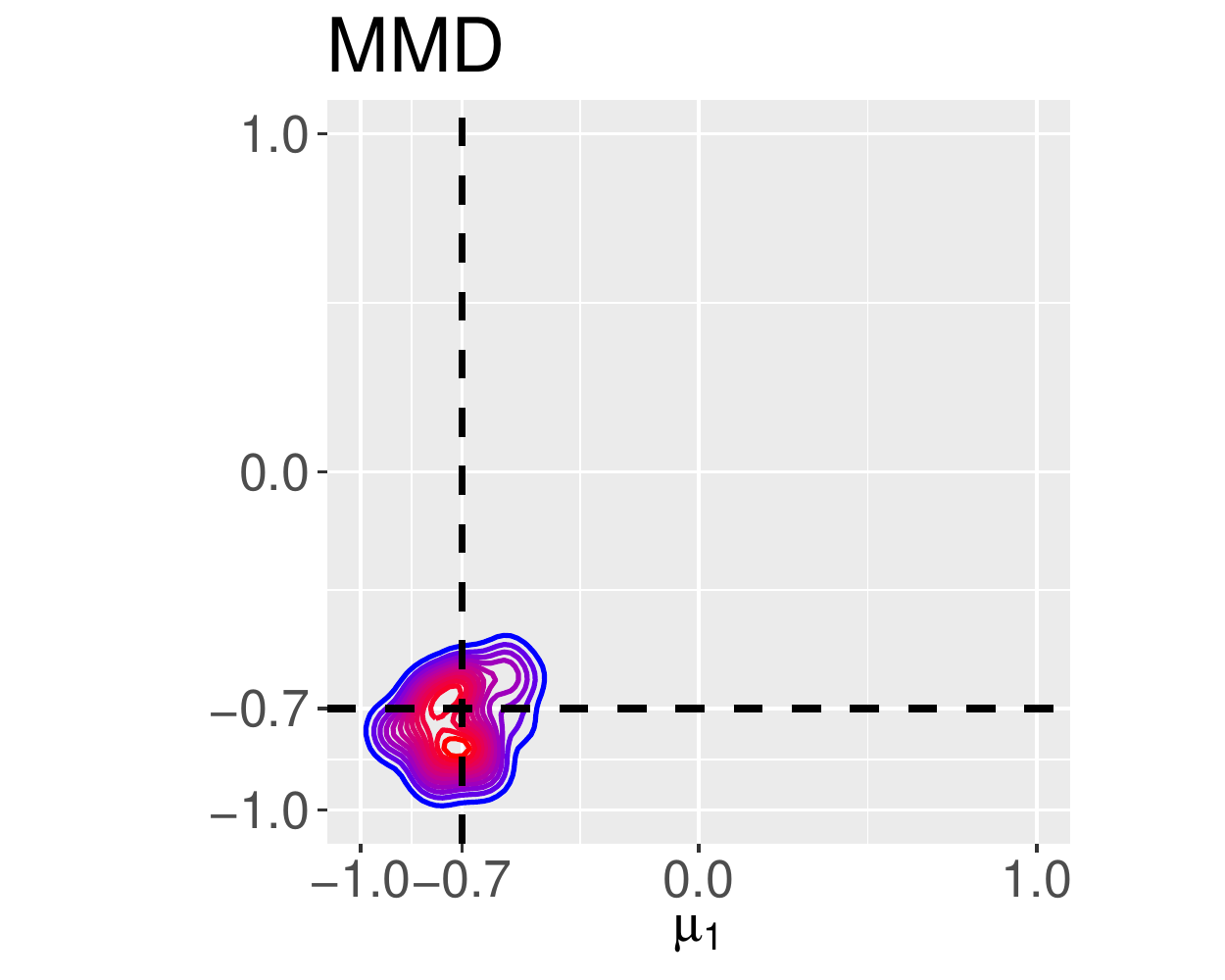}
    \caption{Marginal KDEs of the ABC posterior for the mean parameters $\bm{\mu}_0$ (top row) and $\bm{\mu}_1$ (bottom row) of the bivariate Gaussian mixture model~\eqref{eq: Gaussian_mixture}. The intersections of black dashed lines indicate the positions of the population means.}
    \label{fig: Gaussian mixture}
\end{figure*}

\subsection{Moving-average model of order 2} \label{sec:MA2}

The moving-average model of order $q$, MA($q$), 
is a stochastic process $\{Y_t\}_{t\in\mathbb{N}^{\ast}}$ defined as
\begin{equation*}
    Y_t = Z_t + \sum_{i=1}^q \theta_i Z_{t-i}\text{,}
\end{equation*}
with $\{Z_t\}_{t\in\mathbb{Z}}$ being a sequence of unobserved noise error terms. 
\cite{Jiang2018} used a  MA$(2)$ model for their benchmarking; 
namely $Y_t = Z_t + \theta_1 Z_{t-1} + \theta_2 Z_{t-2},\ t\in[D]$. Each observation $\bm{Y}$ corresponds to a time series of length $D$. Here, 
we use the same model as that proposed in~\cite{Jiang2018}, where $Z_t$ follows the
Student-$t$ distribution with $5$ degrees of freedom, and $D = 10$. The priors 
on the model parameters $\theta_1$ and $\theta_2$ are taken to be uniform, 
that is, $\theta_1 \sim \text{Unif}(-2,2)$ and $\theta_2 \sim \text{Unif}(-1,1)$. 
We performed ABC using $n=200$ samples generated from a model 
with the true parameter values $(\theta_1,\theta_2) = (0.6, 0.2)$. 
A KDE of the ABC joint posterior distribution of $(\theta_1,\theta_2)$ 
is displayed in Figure~\ref{fig: MA2}.

\begin{figure*}[t]
    \centering
    \includegraphics[width=6cm]{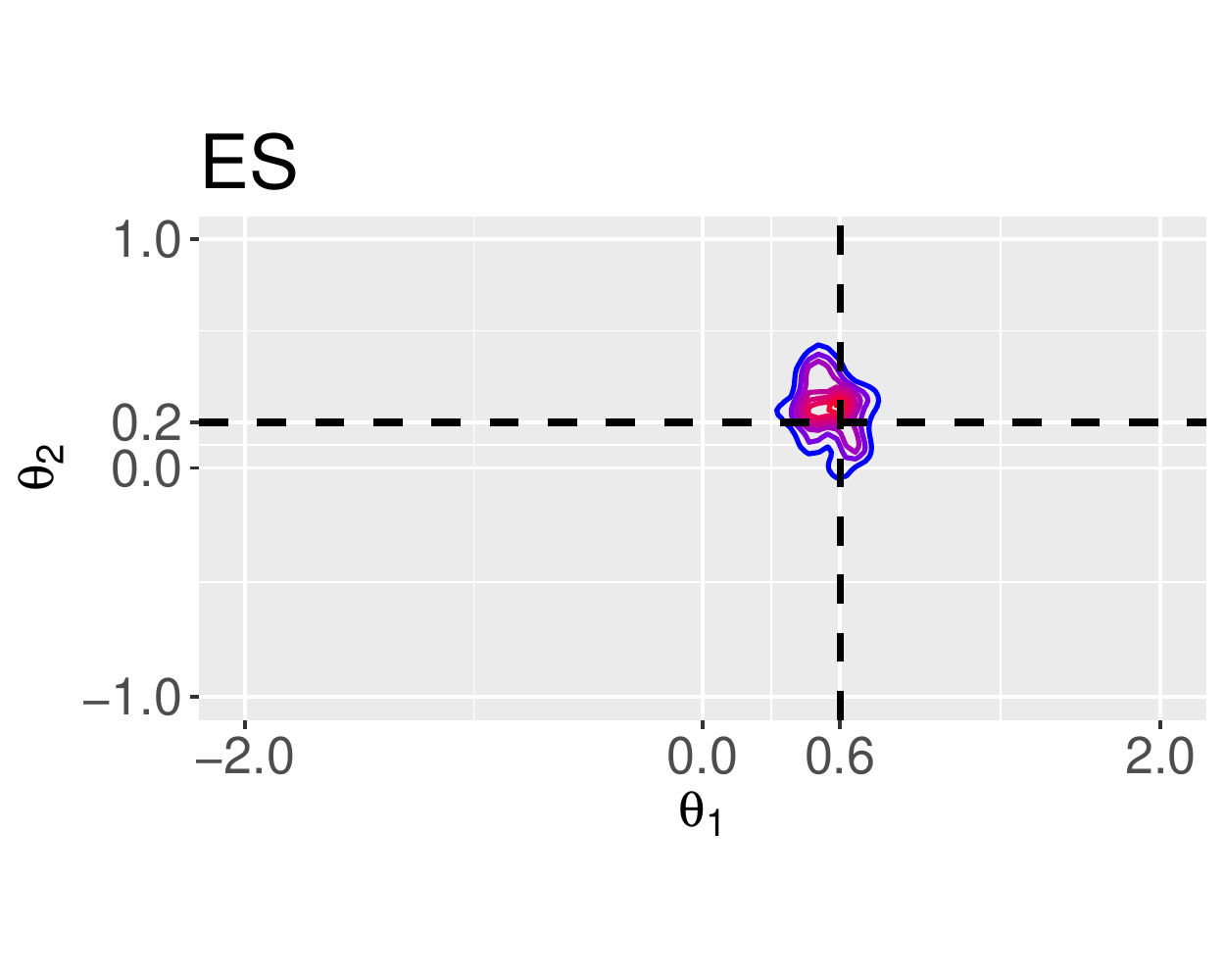}
    \vspace{-1.5cm}
    \includegraphics[width=6cm]{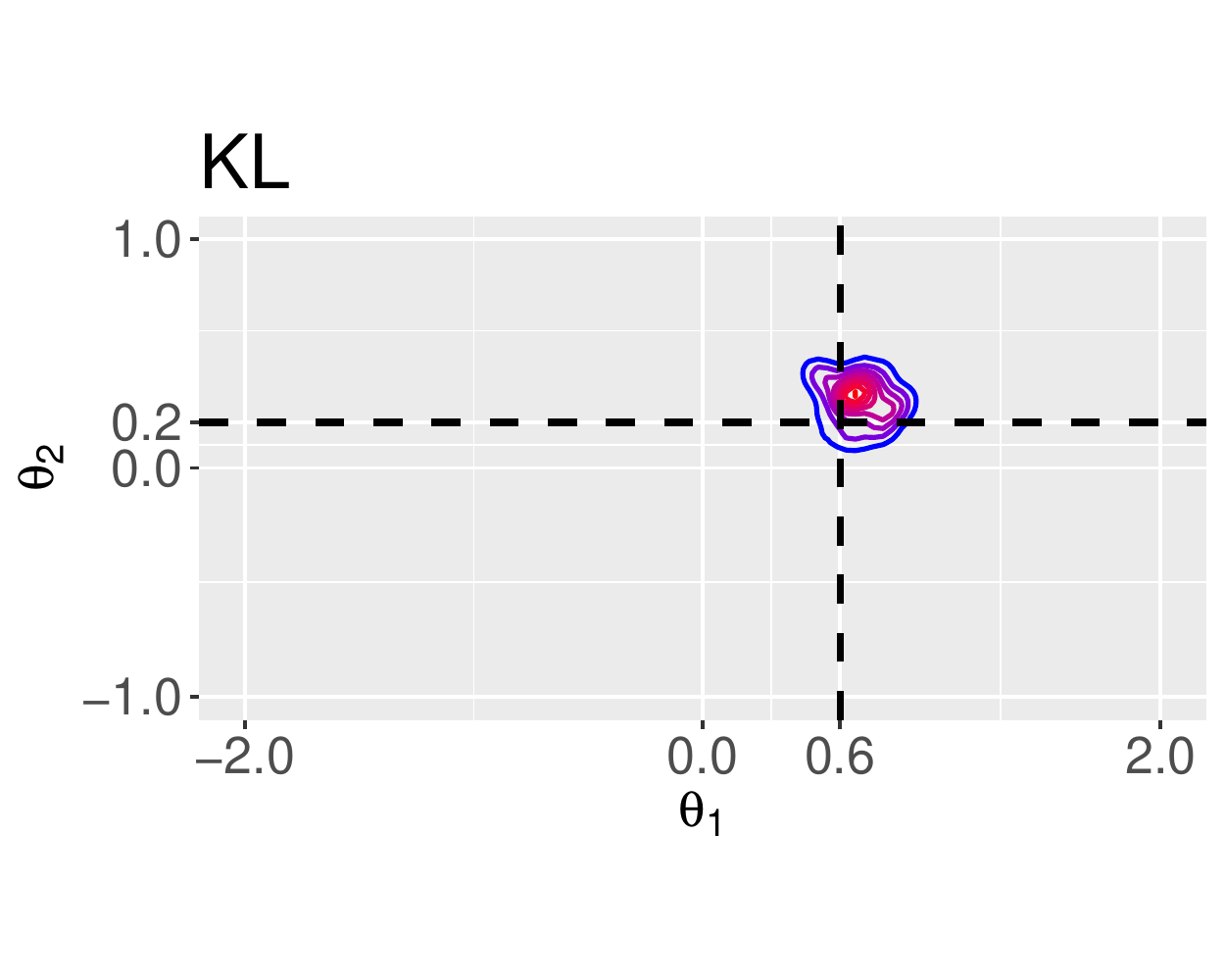}
    \vspace{-1.5cm}
    \includegraphics[width=6cm]{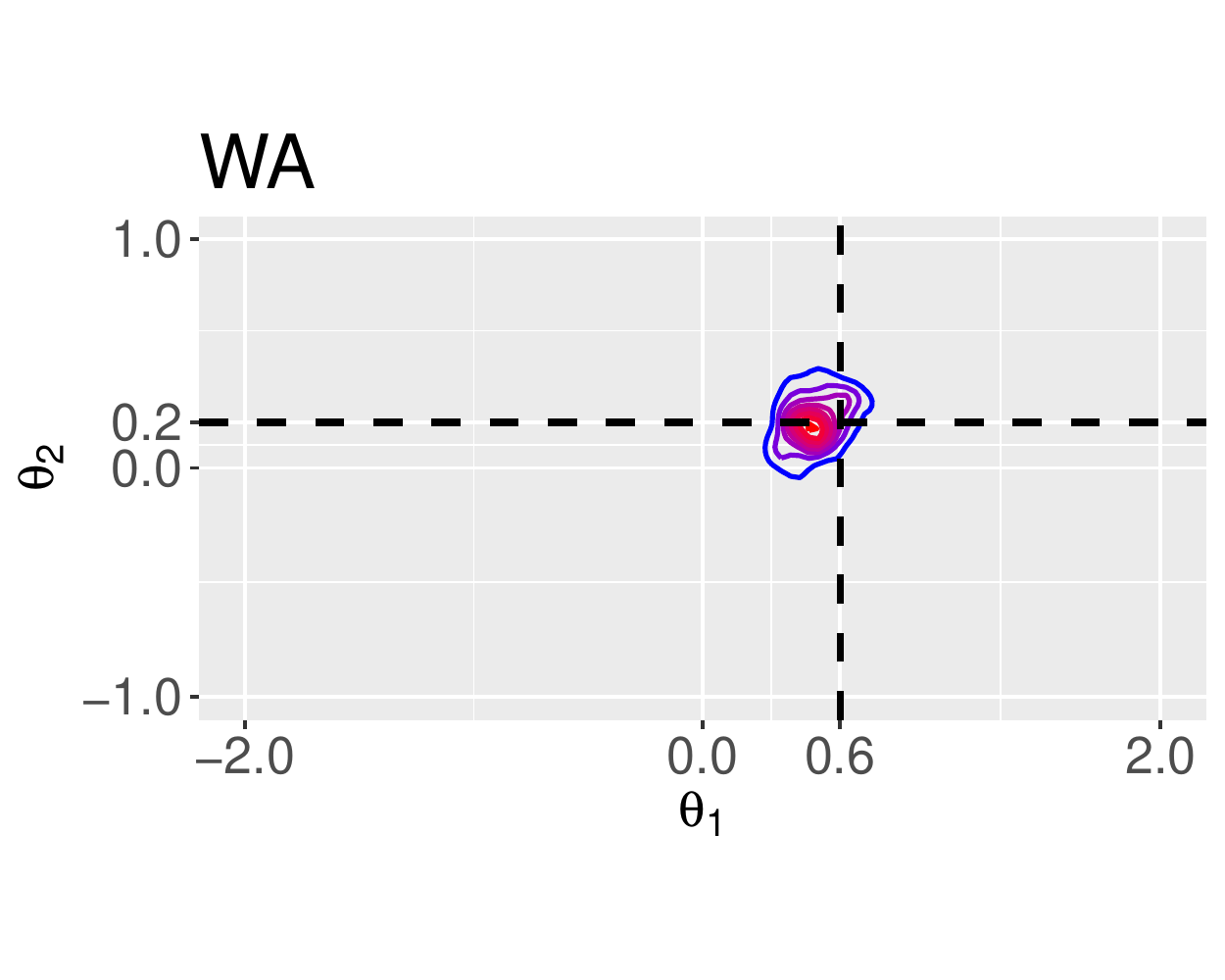}
    \vspace{0.2cm}
    \includegraphics[width=6cm]{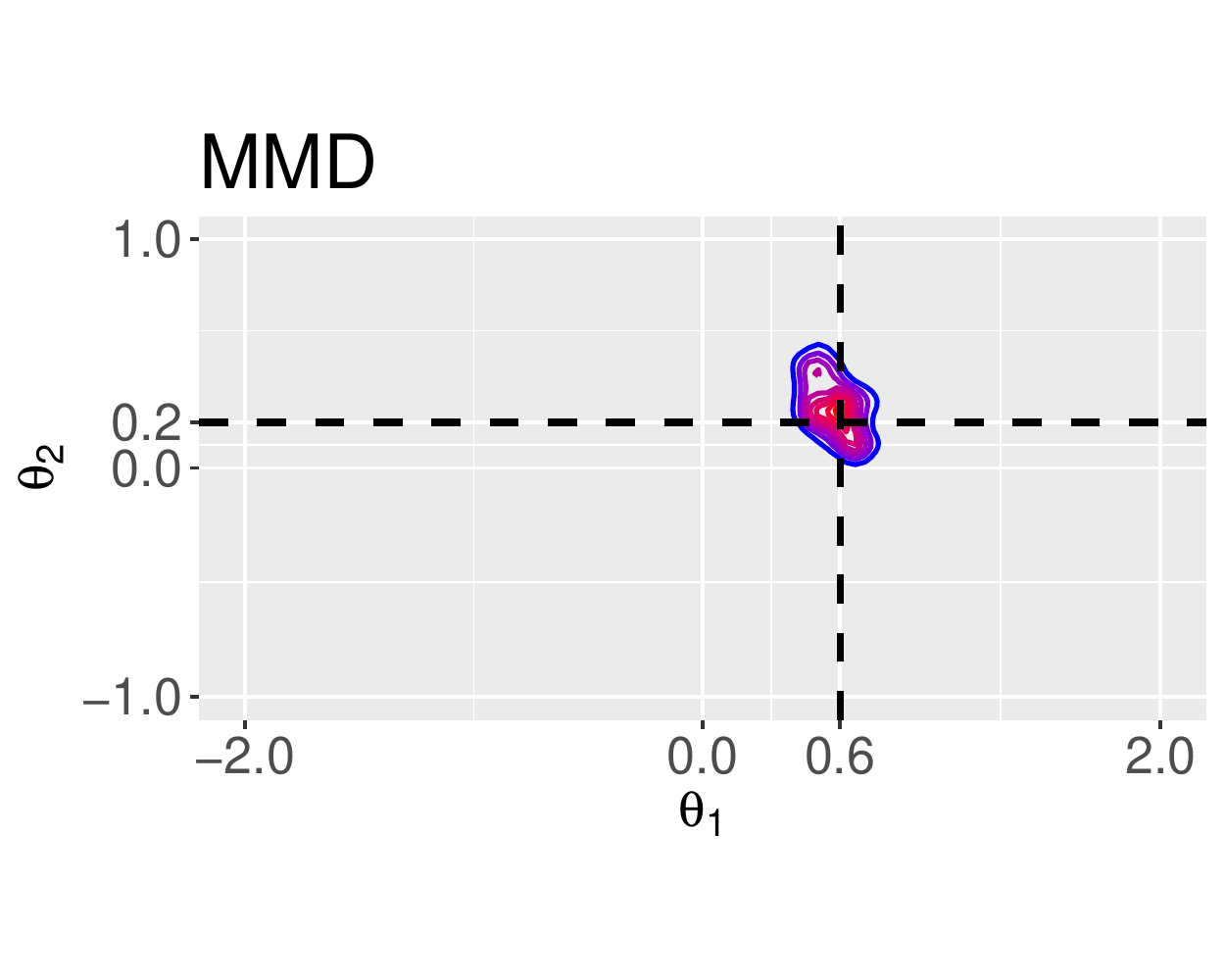}
    \vspace{0.8cm}
    \caption{KDE of the ABC posterior for the parameters 
    $(\theta_1,\theta_2)$ of the MA($2$) model experiment. 
    The intersections of black dashed lines indicate the true parameter values.}
    \label{fig: MA2}
\end{figure*}

\subsection{Bivariate beta model} \label{sec:BBM}

The bivariate beta model  proposed by \cite{Crackel2017} is defined with five positive parameters $\theta_1,\ldots,\theta_5$ by letting
\begin{equation}
    V_1 = \dfrac{U_1 + U_3}{U_5 + U_4}\text{, and }
    V_2 = \dfrac{U_2 + U_4}{U_5 + U_3}\text{,}
\end{equation}
where $U_i \sim \text{Gamma}(\theta_i, 1)$, for $i\in[5]$, and 
setting $Z_1 = V_1 /(1 + V_1 )$ and $Z_2 = V_2 /(1 + V_2 )$. The bivariate random variable $\bm{Z}^\top=(Z_1, Z_2)$ has marginal laws
$Z_1 \sim \mathrm{Beta}(\theta_1 + \theta_3, 
\theta_5 + \theta_4)$ and 
$Z_2 \sim \mathrm{Beta}(\theta_2 + \theta_4, 
\theta_5 + \theta_3)$. 
We performed ABC using samples of size $n =  500$, which 
are generated from a DGP with true parameter values 
$(\theta_1, \theta_2, \theta_3, \theta_4, \theta_5) = (1, 1, 1, 1, 1)$. 
The prior on each of the model parameters is taken to be independent $\mathrm{Unif}(0, 5)$. 
KDEs of the marginal ABC  posterior distributions of parameters $\theta_1, \theta_2, \theta_3, \theta_4$ and  $\theta_5$ are  displayed in Figure~\ref{fig: bbm}.

\begin{figure*}[t]
    \centering
    \includegraphics[trim={0cm 0cm 2.8cm 0cm},clip,height=5.5cm]{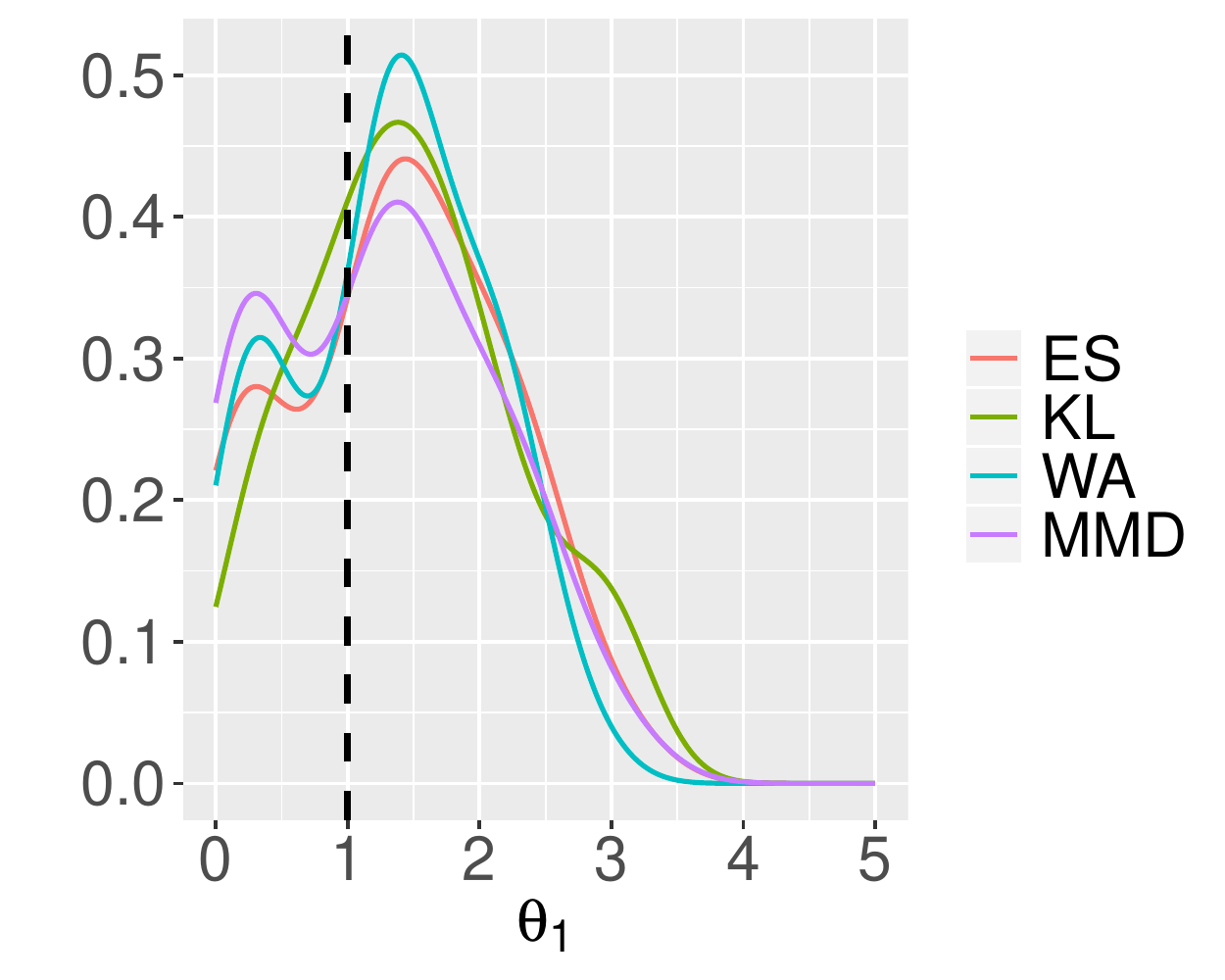}
    \includegraphics[trim={0cm 0cm 2.8cm 0cm},clip,height=5.5cm]{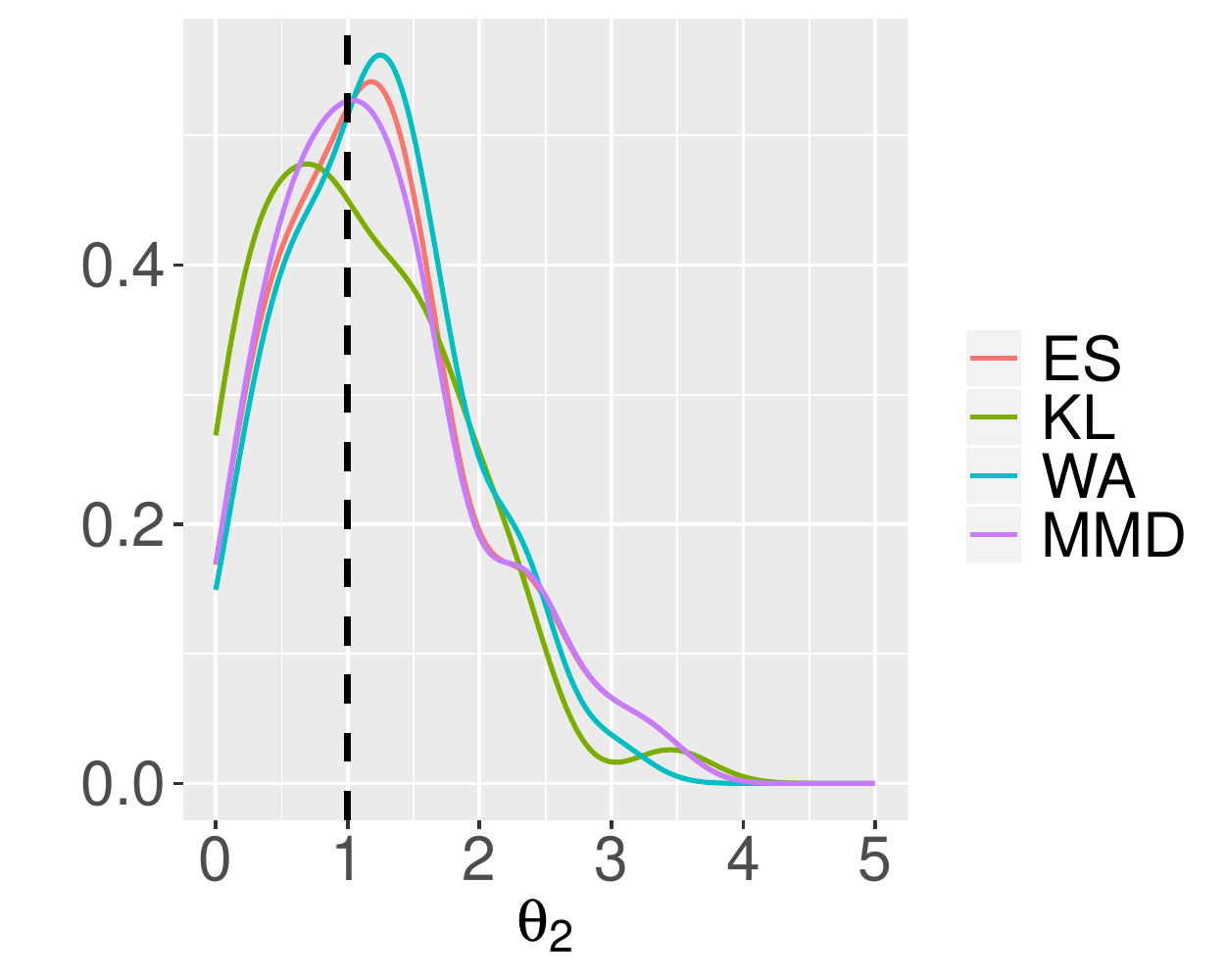}
    \includegraphics[trim={0cm 0cm 2.8cm 0cm},clip,height=5.5cm]{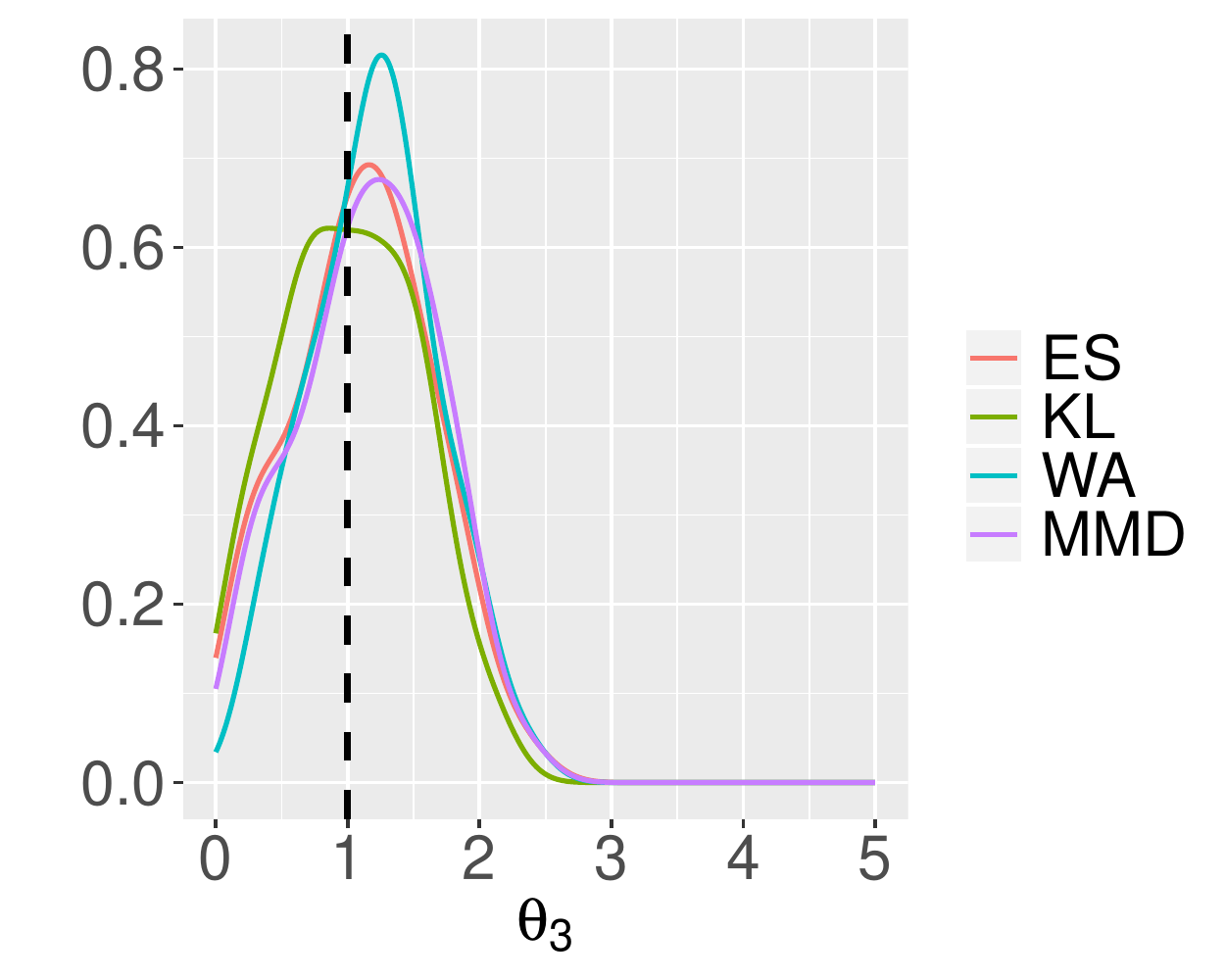}
    \includegraphics[trim={0cm 0cm 2.8cm 0cm},clip,height=5.5cm]{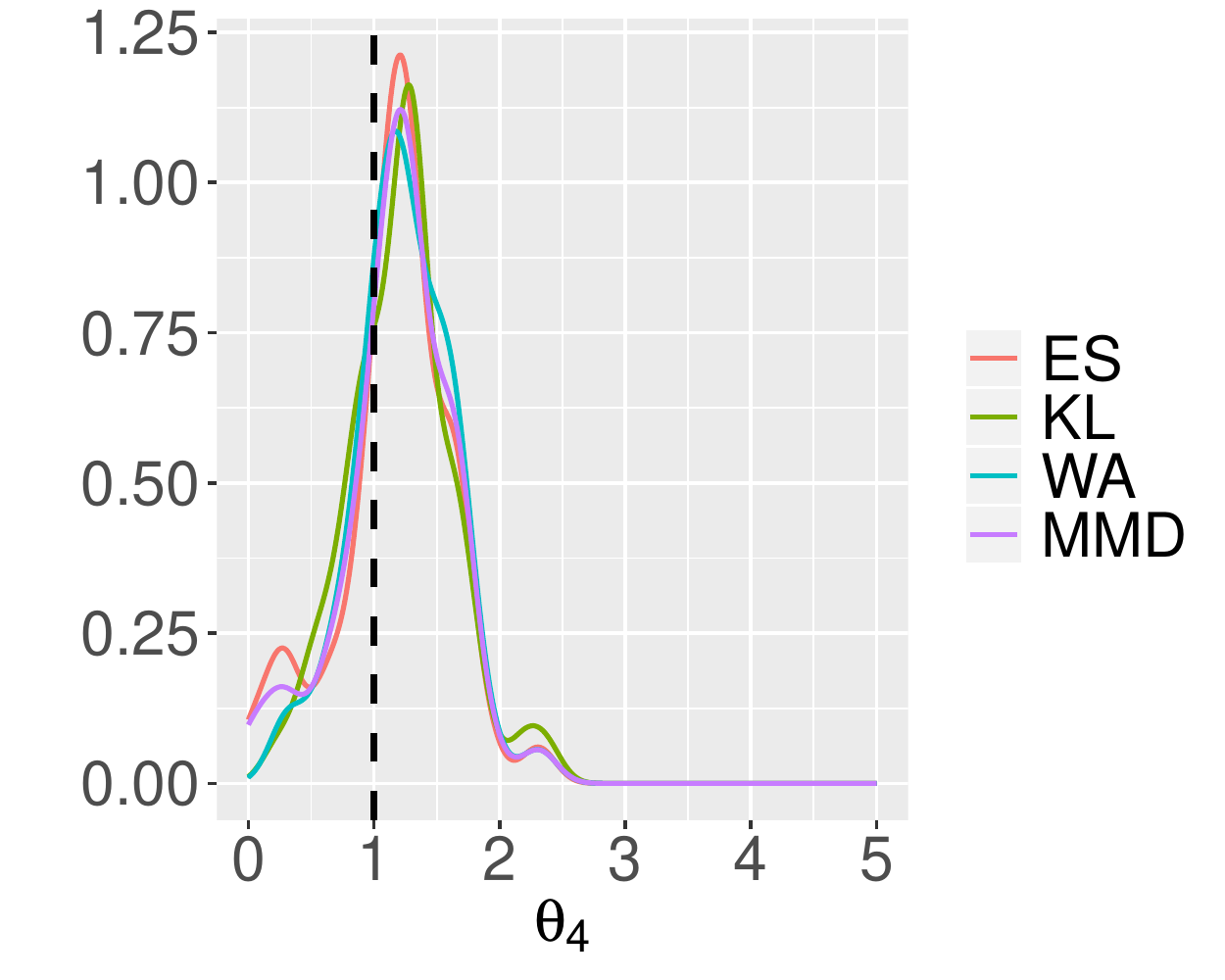}
    \includegraphics[trim={0cm 0cm 0.0cm 0cm},clip,height=5.5cm]{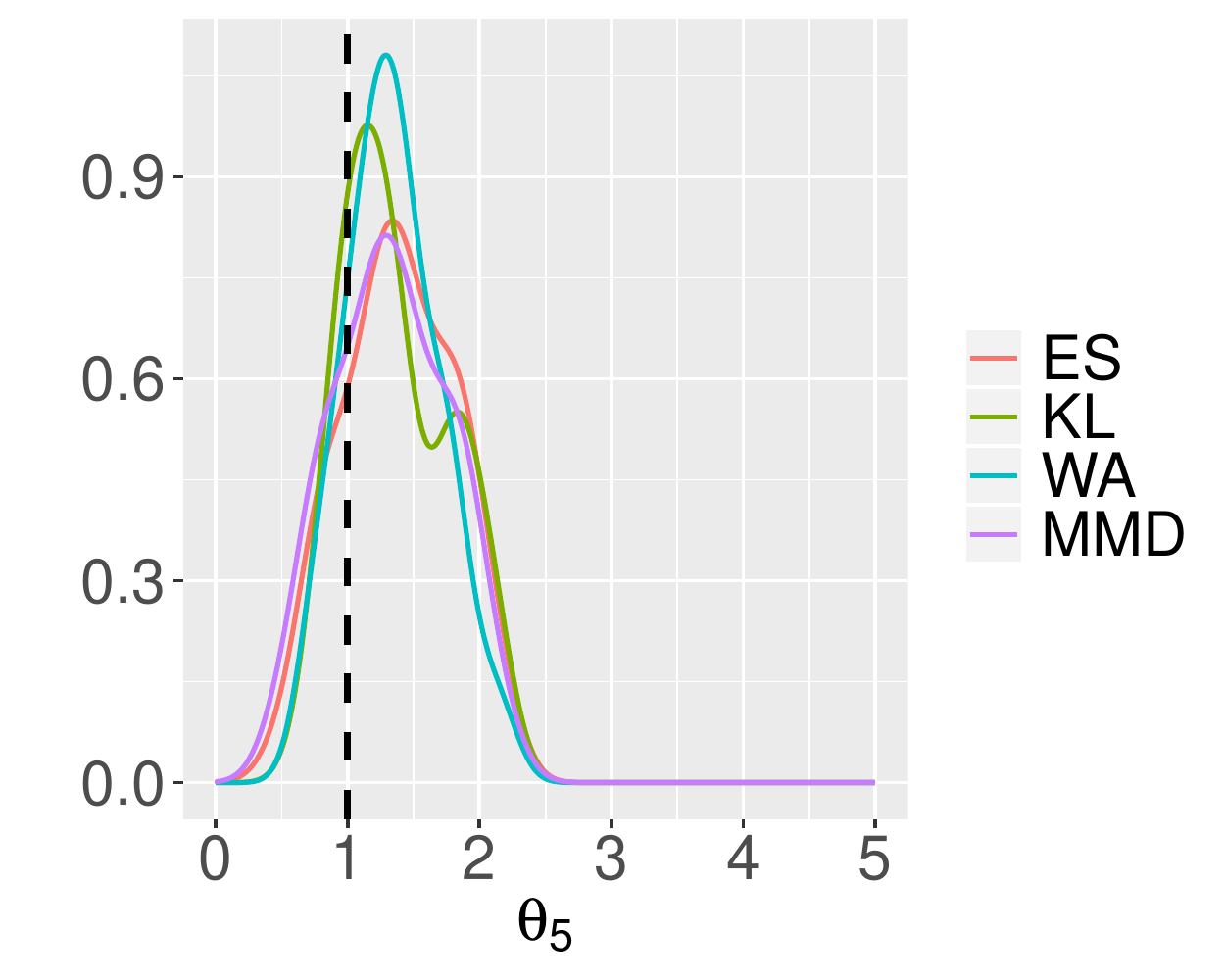}
    \caption{Marginal KDEs of the ABC posterior for the parameters 
    $\theta_1,\ldots,\theta_5$ for the bivariate beta model. 
    The black dashed lines indicate the true parameter values.}
    \label{fig: bbm}
\end{figure*}

\subsection[Multivariate g-and-k distribution]{Multivariate \MakeLowercase{g}-and-\MakeLowercase{k} distribution} \label{sec:gandk}

A univariate $g$-and-$k$ distribution can be defined 
via its quantile function \citep{DROVANDI20112541}:
\begin{equation}\label{eq: g-and-k_1D}
    F^{-1}(x) = A + B \left[1 + 0.8\dfrac{1-\exp(-g \times z_x)}{1+\exp(-g \times z_x)}\right] 
                \left(1+z_x^2\right)^k z_x, 
\end{equation}
where parameters $(A, B, g, k)$ respectively relate to location, scale, 
skewness, and kurtosis. Here, $z_x$ is the $x$th quantile of the standard 
normal distribution. Given a set of 
parameters $(A, B, g, k)$, it is easy to simulate $D$ observations of a DGP 
with quantile function~(\ref{eq: g-and-k_1D}), by generating a sequence of 
IID sample $\{Z_i\}_{i=1}^D$, where $Z_i \sim \mathcal{N}(0, 1)$, for $i\in[D]$. 

A so-called $D$-dimensional $g$-and-$k$ DGP can instead be defined by applying 
the quantile function~(\ref{eq: g-and-k_1D}) to each of the $D$ elements of a 
multivariate normal vector $\bm{Z}^\top=(Z_1, ..., Z_D) \sim \mathcal{N}(\bm{0}, \bm{\Sigma})$, 
where $\bm{\Sigma}$ is a covariance matrix. In our experiment, we use a 
5-dimensional $g$-and-$k$ model with the same covariance matrix and parameter 
values for $(A,B,g,k)$ as that considered by~\cite{Jiang2018}. That is, we generate samples of size 
$n =  200$ from a $g$-and-$k$ DGP with the true parameter values 
$(A,B,g,k)=(3,1,2,0.5)$ and the covariance matrix
\[
\bm{\Sigma}=\left[\begin{array}{ccccc}
1 & \rho & 0 & 0 & 0\\
\rho & 1 & \rho & 0 & 0\\
0 & \rho & 1 & \rho & 0\\
0 & 0 & \rho & 1 & \rho\\
0 & 0 & 0 & \rho & 1
\end{array}\right]\text{,}
\]
where $\rho=-0.3$. 
The prior on the model parameters  $A,B,g,k$ is taken to be independent $\mathrm{Unif}(0, 4)$, while $\rho$ is independently assigned a $\mathrm{Unif}(-0.5, 0.5)$ prior. 
KDEs of the marginal ABC  posterior distributions of parameters $A, B, g, k$ and $\rho$ are displayed in 
Figure~\ref{fig: gandk}.

\begin{figure*}[t]
    \centering
    \includegraphics[trim={0cm 0cm 2.8cm 0cm},clip,height=5.5cm]{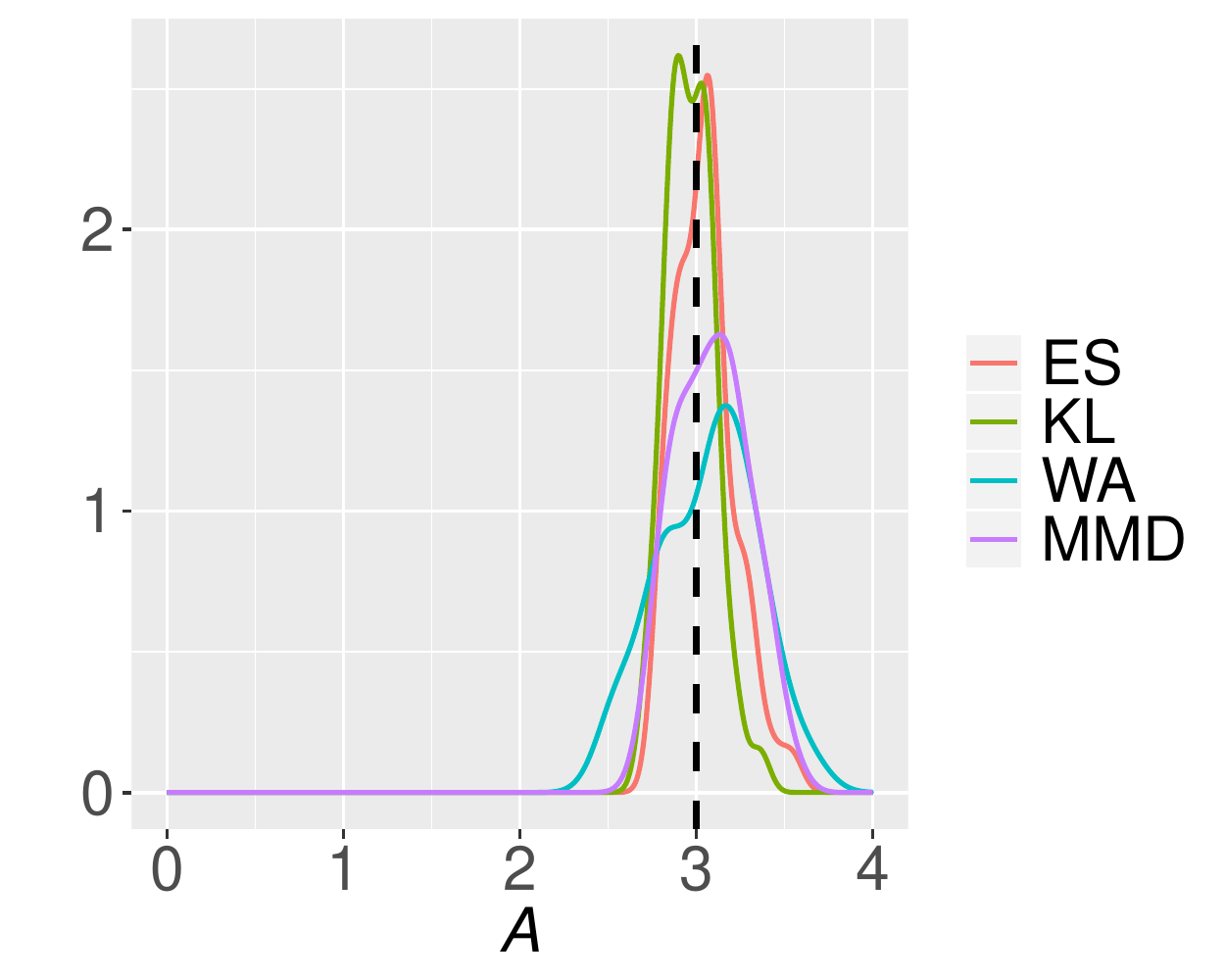}
    \includegraphics[trim={0cm 0cm 2.8cm 0cm},clip,height=5.5cm]{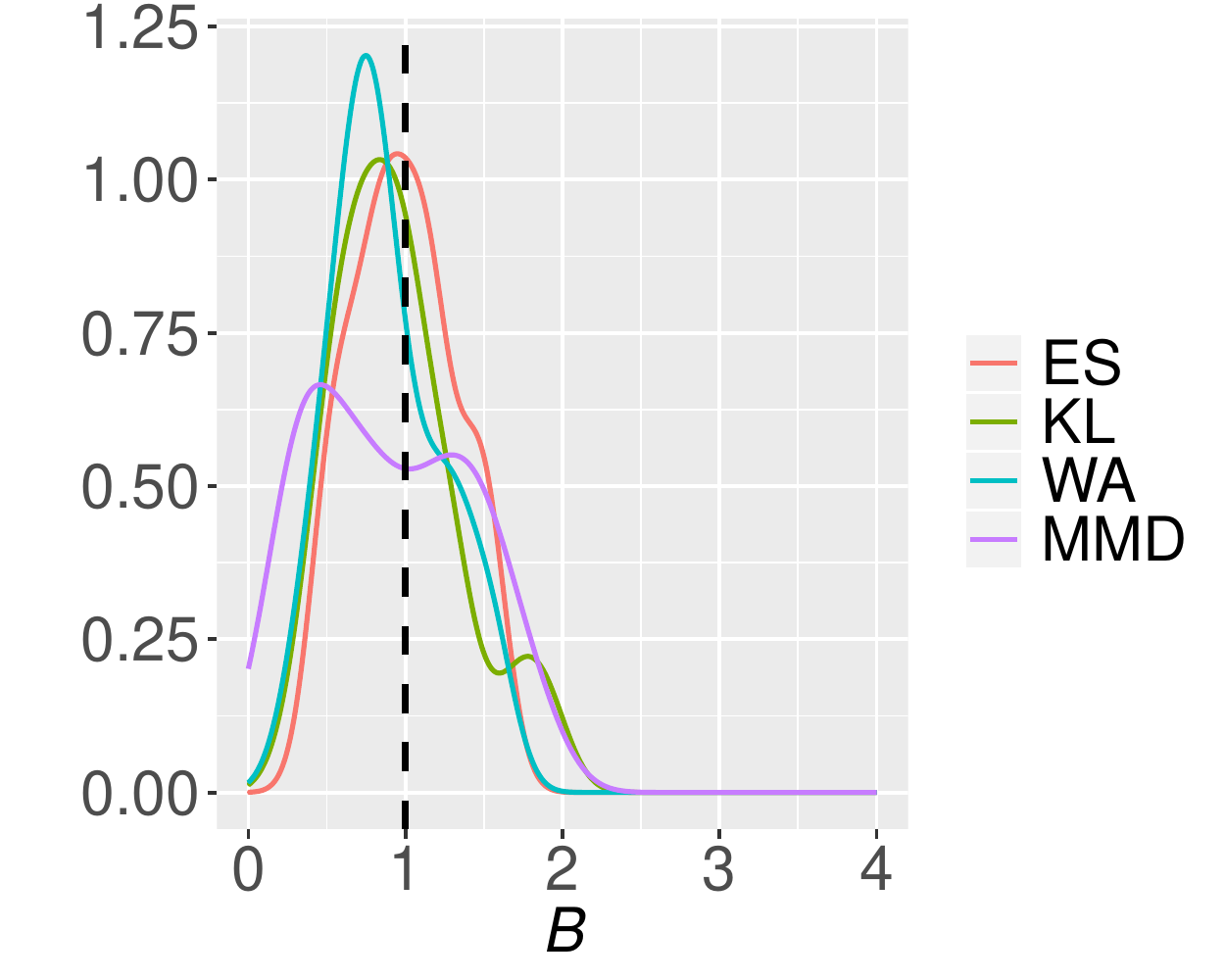}
    \includegraphics[trim={0cm 0cm 2.8cm 0cm},clip,height=5.5cm]{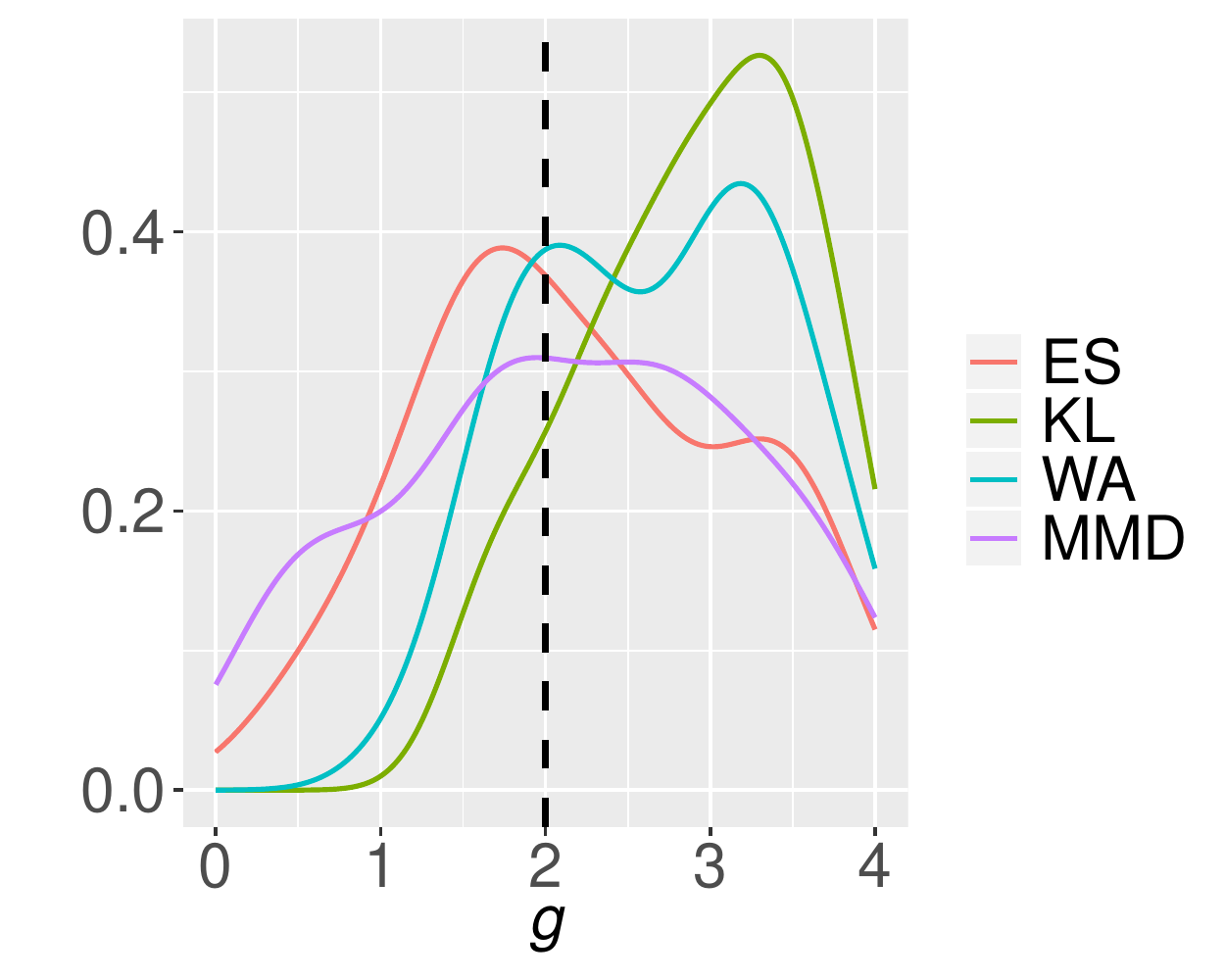}
    \includegraphics[trim={0cm 0cm 2.8cm 0cm},clip,height=5.5cm]{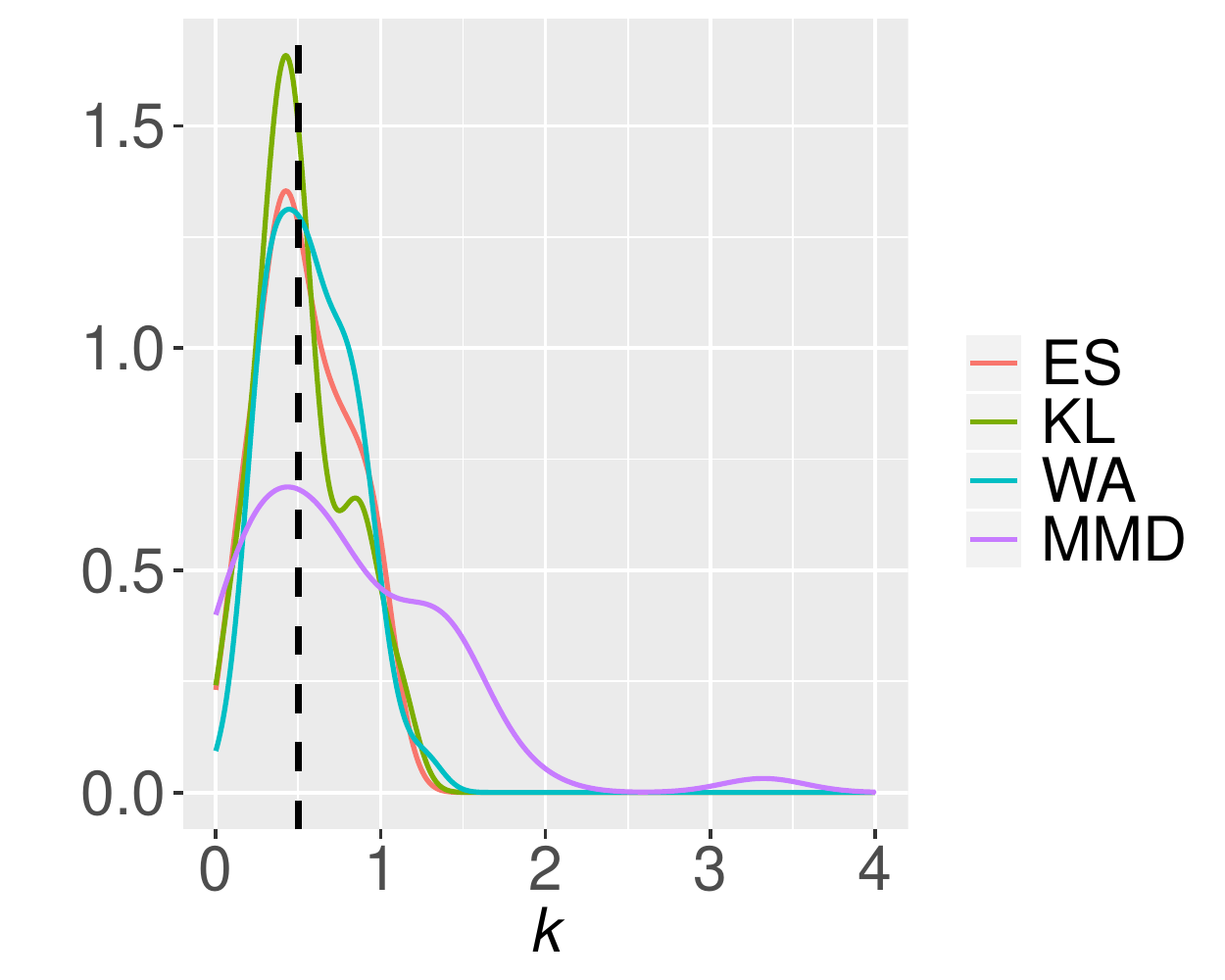}
    \includegraphics[trim={0cm 0cm 0.0cm 0cm},clip,height=5.5cm]{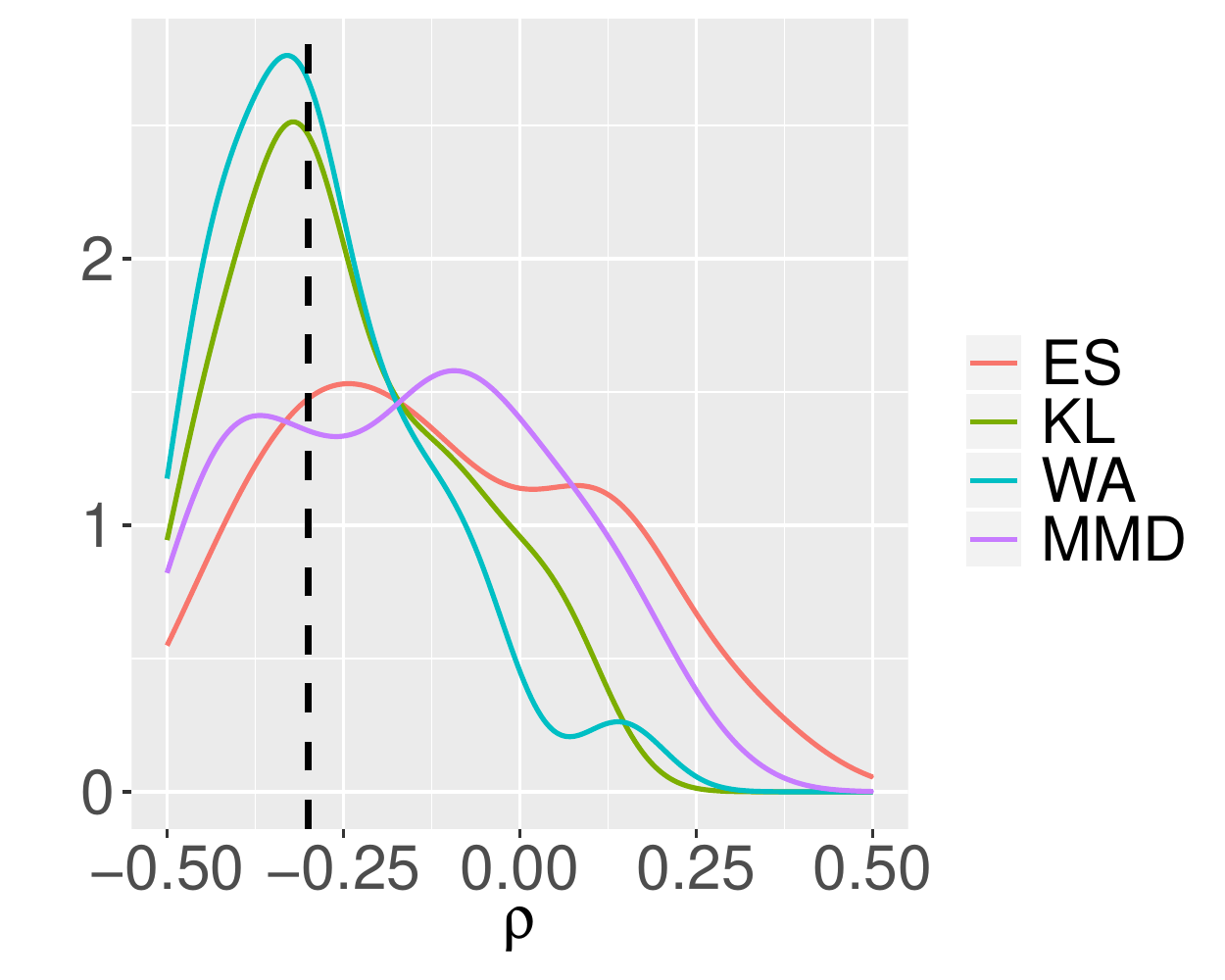}
    \caption{Marginal KDEs of the ABC posterior for the parameters 
    $A,B,g,k$ and $\rho$ of the $g$-and-$k$ model. 
    The black dashed lines indicate the true parameter values.}
    \label{fig: gandk}
\end{figure*}

\subsection{Discussion of the results and performance} \label{sec:illustration-discussion}

For each of the four  experiments and each parameter, we computed the posterior mean $\hat \theta_{\text{mean}}$, posterior median $\hat \theta_{\text{med}}$, mean absolute error
and mean squared error defined by
\begin{equation*}
     \text{MAE}=\frac{1}{M}\sum_{k=1}^M |\theta_k-\theta_0 |\text{, and} \quad 
     \text{MSE} = \frac{1}{M}\sum_{k=1}^M |\theta_k-\theta_0 |^2,
\end{equation*}
where $\left\{\theta_k\right\}_{k=1} ^M$ denotes the pseudo-posterior sample and $\theta_0$ denotes the true parameter. Here $M=50$ since $N=10^5$ and $\epsilon$ is chosen as to retain $0.05\%$ of the samples.  Each experiment was replicated ten times by  keeping the same fixed (true) values for the parameters and by sampling new observed data each of the ten times. The estimated quantities $\hat \theta_{\text{mean}}$, $\hat \theta_{\text{med}}$, and errors MAE  and $\text{RMSE}=\text{MSE}^{1/2}$ were then averaged over the ten replications, and are reported along with standard deviations $\sigma(\cdot)$ in columns associated with each estimator  and true values $\theta_0$ for each parameter in Tables \ref{tab:performances_bgm}, \ref{tab:performances_ma2}, \ref{tab:performances_bbm} and 
\ref{tab:performances_gandk}. 

Upon inspection, Tables \ref{tab:performances_bgm}, \ref{tab:performances_ma2}, \ref{tab:performances_bbm} and 
\ref{tab:performances_gandk} showed some advantage in performance from WA on the bivariate Gaussian mixtures, some advantage from the MMD on the bivariate beta model, and some advantage from the ES on the $g$-and-$k$ model, while multiple methods are required to make the best inference in the case of the MA$(2)$ experiment. When we further take into account 
the standard deviations of the estimators, we observe that all four data discrepancy measures essentially perform comparatively well across the four experimental models. Thus, we may conclude that there is no universally best performing discrepancy measure. \textcolor{black}{Some considerations are therefore necessary when choosing between discrepancies. The first point of consideration is whether the data $\mathbf{X}_n$ are random variables arising from continuous or discrete measures. In the case that the data $\mathbf{X}_n$ arises from a discrete measure, the KL discrepancy measure is not applicable, since it is not defined on a set of measure greater than zero. Another consideration regarding the choice of discrepancy measures is the computational complexity of each discrepancy measure, as is summarized in Table~\ref{tab:comput_complexities}.}

\textcolor{black}{From Table~\ref{tab:comput_complexities}, we firstly note that in the case of univariate data, all methods have the same computational complexity, as all of the discrepancy measures amount to comparisons between the order statistics of the observed and simulated data. Computational complexity becomes a greater separating criterion when considering the multivariate setting. In the multivariate case, the KL divergence is clearly faster than the other methods, but as mentioned before, is not applicable for discrete data. The ES and MMD methods share the same order of complexity, $\mathcal{O}((n+m)^{2})$, due to their theoretical equivalence (cf. \cite{Sejdinovic:2013aa}). It is notable that, in general, the computational complexity of the WA discrepancy is of order $\mathcal{O}((n+m)^{5/2}\log(n+m))$, which greater than that of the ES and MMD discrepancies, and is thus a significantly slower method when $n$ and $m$ get large. However, in our numerical results, we have used the $\mathcal{O}((n+m)^{2})$ swapping distance approximation of the WA method, as was considered in \cite{Bernton:2017aa}. Although this approximation is faster than the exact WA discrepancy, it does not converge to the same value, in general, and thus theoretical results regarding the WA discrepancy cannot be directly applied to the approximation (although some theoretical statements are still available). Thus, there is a trade-off regarding theoretical outcomes when using the swap distance approximation.}

\textcolor{black}{We note that in the case when the MMD discrepancy measure is estimated by the V-statistic estimator, much of our theoretical results from Section~\ref{sec:theory}, are applicable with minor modifications, due to the results of \cite{Sejdinovic:2013aa}. Thus, the choice between the ES and the MMD method comes down to a preference for the use of kernels or metrics. A consideration regarding the choice of the MMD discrepancy versus the ES discrepancy is that, to the best of our knowledge, a comparable result to \eqref{eq: ES limit} does not exist for any common kernel choice.}

As an alternative to choosing one of the assessed discrepancy measures, one may also consider some kind of averaging over the results of the different discrepancy measures. We have not committed to an investigation of such methodologies and leave it as a future research direction.

\textcolor{black}{Running times (on a MacBook Pro 3,1 GHz) for the ES, KL, MMD and WA distance computations for $10^5$ ABC replications, in the four models considered in the simulations, for varying sample sizes $n$, and with $m=n$, 
are reported in Figure~\ref{fig:times}. ES is uniformly much faster than the other approaches for small samples sizes, up to the value of $n=m=50$, where it is performing as fast as KL. For sample sizes larger than  $n=m=50$, KL is fastest. Overall, MMD and WA are slower than ES and KL. 
}

\begin{table*}
    \centering
    \caption{Estimation performance for bivariate Gaussian mixtures (Section~\ref{sec:BGM}). The best results in each column is highlighted in boldface.}
    \label{tab:performances_bgm}
    \scalebox{0.9}{
    \begin{tabular}{cccccccccc}
    \toprule
    && $\hat \theta_{\text{mean}}$ & $\sigma(\hat \theta_{\text{mean}})$ 
    & $\hat \theta_{\text{med}}$ & $\sigma(\hat \theta_{\text{med}})$  
    & MAE & $\sigma(\mathrm{MAE})$ & RMSE &  $\sigma(\mathrm{RMSE})$ \\
    \toprule
    \multirow{4}{*}{ $\mu_{00}=0.7$} 
      & ES  & 0.594 & 0.045 & 0.607 & 0.063 & 0.215 & 0.030 & 0.283 & 0.055 \\ 
      & KL & 0.648 & 0.039 & 0.666 & 0.048 & 0.165 & 0.016 & 0.205 & 0.026 \\ 
      & \textbf{WA} & \textbf{0.675} & 0.035 & \textbf{0.682} & 0.043 & \textbf{0.152} & 0.020 & \textbf{0.181} & 0.021 \\ 
      & MMD & 0.564 & 0.079 & 0.582 & 0.076 & 0.234 & 0.054 & 0.311 & 0.101 \\ 
      \hline
    \multirow{4}{*}{$\mu_{01}=0.7$} 
      & ES  & 0.587 & 0.063 & 0.613 & 0.059 & 0.215 & 0.038 & 0.282 & 0.069 \\ 
      & KL & 0.651 & 0.042 & 0.667 & 0.061 & 0.169 & 0.022 & 0.210 & 0.027 \\ 
      & \textbf{WA} & \textbf{0.655} & 0.050 & \textbf{0.669} & 0.047 & \textbf{0.152} & 0.015 & \textbf{0.187} & 0.019 \\ 
      & MMD & 0.559 & 0.076 & 0.598 & 0.075 & 0.235 & 0.049 & 0.313 & 0.092 \\  
      \hline
    \multirow{4}{*}{$\mu_{10}=-0.7$} 
      & \textbf{ES} & \textbf{-0.699} & 0.046 & -0.716 & 0.040 & 1.401 & 0.043 & 1.412 & 0.039 \\ 
      & KL & -0.709 & 0.029 & -0.712 & 0.035 & 1.409 & 0.029 & 1.415 & 0.029 \\ 
      & \textbf{WA} & \textbf{-0.699} & 0.030 & \textbf{-0.704} & 0.037 & \textbf{1.399} & 0.030 & \textbf{1.404} & 0.030 \\ 
      & MMD & -0.709 & 0.054 & -0.731 & 0.036 & 1.411 & 0.051 & 1.422 & 0.038 \\ 
      \hline
    \multirow{4}{*}{$\mu_{11}=-0.7$} 
      & \textbf{ES}  & \textbf{-0.696} & 0.058 & -0.712 & 0.043 & 1.396 & 0.058 & 1.407 & 0.049 \\ 
      & \textbf{KL} & -0.711 & 0.047 & \textbf{-0.704} & 0.057 & 1.411 & 0.047 & 1.416 & 0.047 \\ 
      & \textbf{WA} & -0.695 & 0.043 & -0.695 & 0.053 & \textbf{1.395} & 0.043 & \textbf{1.401} & 0.043 \\
      & MMD & -0.711 & 0.066 & -0.726 & 0.046 & 1.411 & 0.066 & 1.424 & 0.052 \\ 
      \bottomrule
    \end{tabular}
    }
\end{table*}%

\begin{table*}
    \centering
    \caption{Estimation performance for the MA($2$) model  (Section~\ref{sec:MA2}). The best results in each column is highlighted in boldface.}
    \label{tab:performances_ma2}
    \scalebox{0.9}{
    \begin{tabular}{cccccccccc}
    \toprule
    && $\hat \theta_{\text{mean}}$ & $\sigma(\hat \theta_{\text{mean}})$ 
    & $\hat \theta_{\text{med}}$ & $\sigma(\hat \theta_{\text{med}})$  
    & MAE & $\sigma(\mathrm{MAE})$ & RMSE &  $\sigma(\mathrm{RMSE})$ \\
    \toprule
    \multirow{4}{*}{ $\theta_1=0.6$} 
      & ES  & 0.569 & 0.042 & 0.570 & 0.045 & 0.083 & 0.015 & 0.100 & 0.017 \\
      & KL & 0.664 & 0.028 & 0.658 & 0.031 & 0.106 & 0.017 & 0.132 & 0.019 \\
      & WA & 0.509 & 0.033 & 0.505 & 0.038 & 0.112 & 0.022 & 0.133 & 0.026 \\ 
      & \textbf{MMD} & \textbf{0.583} & 0.044 & \textbf{0.586} & 0.048 & \textbf{0.079} & 0.013 & \textbf{0.096} & 0.015 \\
      \hline
    \multirow{4}{*}{$\theta_2=0.2$} 
      & ES & 0.215 & 0.035 & 0.219 & 0.035 & 0.111 & 0.015 & 0.135 & 0.019 \\ 
      & KL & 0.274 & 0.023 & 0.280 & 0.027 & 0.110 & 0.014 & 0.134 & 0.014 \\ 
      & \textbf{WA} & \textbf{0.205} & 0.025 & \textbf{0.207} & 0.030 & \textbf{0.090} & 0.029 & \textbf{0.112} & 0.034 \\ 
      & MMD & 0.220 & 0.037 & 0.220 & 0.036 & 0.108 & 0.010 & 0.132 & 0.012 \\
      \bottomrule
    \end{tabular}
    }
\end{table*}%

\begin{table*}
    \centering
    \caption{Estimation performance for the bivariate beta model (Section~\ref{sec:BBM}). The best results in each column is highlighted in boldface.}
    \label{tab:performances_bbm}
    \scalebox{0.9}{
    \begin{tabular}{cccccccccc}
    \toprule
    && $\hat \theta_{\text{mean}}$ & $\sigma(\hat \theta_{\text{mean}})$ 
    & $\hat \theta_{\text{med}}$ & $\sigma(\hat \theta_{\text{med}})$  
    & MAE & $\sigma(\mathrm{MAE})$ & RMSE &  $\sigma(\mathrm{RMSE})$ \\
    \toprule
    \multirow{4}{*}{ $\theta_1=1.0$} 
      & ES  & 1.299 & 0.223 & 1.189 & 0.264 & 0.713 & 0.130 & 0.885 & 0.165 \\ 
      & KL & 1.389 & 0.190 & 1.333 & 0.165 & 0.696 & 0.151 & 0.877 & 0.205 \\
      & \textbf{WA} & \textbf{1.286} & 0.220 & 1.193 & 0.265 & 0.672 & 0.128 & \textbf{0.828} & 0.153 \\
      & \textbf{MMD} & 1.229 & 0.188 & \textbf{1.143} & 0.241 & \textbf{0.676} & 0.092 & 0.836 & 0.121 \\
      \hline
     \multirow{4}{*}{ $\theta_2=1.0$} 
      & ES & 1.362 & 0.185 & 1.290 & 0.237 & 0.716 & 0.118 & 0.904 & 0.131 \\ 
      & \textbf{KL} & \textbf{1.235} & 0.152 & \textbf{1.153} & 0.170 & \textbf{0.588} & 0.070 & \textbf{0.745} & 0.097 \\ 
      & WA & 1.292 & 0.196 & 1.240 & 0.241 & 0.657 & 0.114 & 0.817 & 0.139 \\ 
      & MMD & 1.268 & 0.173 & 1.170 & 0.171 & 0.669 & 0.103 & 0.841 & 0.131 \\ 
      \hline
      \multirow{4}{*}{ $\theta_3=1.0$} 
      & ES & 1.170 & 0.132 & 1.183 & 0.157 & 0.459 & 0.045 & 0.552 & 0.049 \\  
      & \textbf{KL} & \textbf{1.083} & 0.100 & \textbf{1.077} & 0.088 & \textbf{0.394} & 0.034 & \textbf{0.496} & 0.045 \\
      & WA & 1.229 & 0.118 & 1.216 & 0.132 & 0.426 & 0.054 & 0.521 & 0.059 \\ 
      & MMD & 1.181 & 0.116 & 1.182 & 0.143 & 0.456 & 0.051 & 0.548 & 0.061 \\ 
      \hline
      \multirow{4}{*}{ $\theta_4=1.0$} 
      & \textbf{ES} & \textbf{1.128} & 0.112 & 1.113 & 0.138 & 0.435 & 0.032 & 0.534 & 0.045 \\ 
      & \textbf{KL} & 1.133 & 0.111 & \textbf{1.086} & 0.135 & \textbf{0.390} & 0.038 & \textbf{0.498} & 0.051 \\ 
      & WA & 1.218 & 0.110 & 1.196 & 0.108 & 0.409 & 0.049 & 0.514 & 0.066 \\ 
      & MMD & 1.150 & 0.098 & 1.133 & 0.130 & 0.423 & 0.041 & 0.518 & 0.049 \\ 
      \hline
      \multirow{4}{*}{ $\theta_5=1.0$} 
      & ES & 1.343 & 0.096 & 1.360 & 0.104 & 0.428 & 0.052 & 0.514 & 0.059 \\
      & KL & 1.300 & 0.087 & 1.250 & 0.065 & 0.384 & 0.040 & 0.491 & 0.061 \\
      & \textbf{WA} & 1.300 & 0.101 & 1.298 & 0.105 & \textbf{0.370} & 0.058 & \textbf{0.446} & 0.066 \\
      & \textbf{MMD} & \textbf{1.258} & 0.115 & \textbf{1.232} & 0.120 & 0.375 & 0.055 & 0.454 & 0.063 \\
      \bottomrule
    \end{tabular}
    }
\end{table*}%

\begin{table*}
    \centering
    \caption{Computational complexities. See discussion in Section~\ref{sec:conclusion}.}
    \label{tab:comput_complexities}
    \begin{tabular}{lcc}
    \toprule
    &Complexity & References\\ \hline 
    Univariate (all methods) & $\mathcal{O}((n+m)\log(n+m))$ & \cite{Jiang2018,Bernton:2017aa,huo2016fast,chaudhuri2019fast}\\
    KL & $\mathcal{O}((n+m)\log(n+m))$  & \cite{Jiang2018} \\
    Multivariate ES/MMD, WA (approx.) & $\mathcal{O}((n+m)^2)$ & \cite{Jiang2018,Bernton:2017aa}\\
    Multivariate WA & $\mathcal{O}((n+m)^{5/2}\log(n+m))$  & \cite{Bernton:2017aa} \\
    \bottomrule
    \end{tabular}
\end{table*}%

\begin{table*}
    \centering
    \caption{Estimation performance for the $g$-and-$k$ distribution (Section~\ref{sec:gandk}). The best results in each column is highlighted in boldface.}
    \label{tab:performances_gandk}
    \scalebox{0.9}{
    \begin{tabular}{cccccccccc}
    \toprule
    && $\hat \theta_{\text{mean}}$ & $\sigma(\hat \theta_{\text{mean}})$ 
    & $\hat \theta_{\text{med}}$ & $\sigma(\hat \theta_{\text{med}})$  
    & MAE & $\sigma(\mathrm{MAE})$ & RMSE &  $\sigma(\mathrm{RMSE})$ \\
    \toprule
    \multirow{4}{*}{ $A=3.0$} 
      & \textbf{ES}  & \textbf{3.024} & 0.044 & \textbf{3.009} & 0.047 & 0.133 & 0.016 & 0.170 & 0.018 \\ 
      & \textbf{KL} & 2.955 & 0.030 & 2.948 & 0.033 & \textbf{0.105} & 0.013 & \textbf{0.128} & 0.013 \\ 
      & WA & 3.043 & 0.045 & 3.052 & 0.067 & 0.232 & 0.020 & 0.277 & 0.020 \\ 
      & MMD & 3.081 & 0.061 & 3.062 & 0.065 & 0.177 & 0.029 & 0.221 & 0.036 \\ 
      \hline
    \multirow{4}{*}{ $B=1.0$} 
      & \textbf{ES} & \textbf{1.046} & 0.062 & \textbf{1.027} & 0.079 & \textbf{0.268} & 0.024 & \textbf{0.322} & 0.029 \\ 
      & KL & 0.918 & 0.071 & 0.885 & 0.068 & 0.313 & 0.026 & 0.375 & 0.029 \\ 
      & WA & 0.894 & 0.127 & 0.869 & 0.136 & 0.277 & 0.044 & 0.334 & 0.045 \\ 
      & MMD & 0.899 & 0.069 & 0.855 & 0.079 & 0.374 & 0.029 & 0.440 & 0.030 \\
      \hline
    \multirow{4}{*}{ $g=2.0$} 
      & ES & 2.289 & 0.101 & 2.264 & 0.210 & 0.872 & 0.098 & 1.026 & 0.091 \\ 
      & KL & 2.993 & 0.080 & 3.046 & 0.121 & 1.043 & 0.070 & 1.193 & 0.066 \\ 
      & \textbf{WA} & 2.581 & 0.101 & 2.599 & 0.147 & \textbf{0.858} & 0.078 & \textbf{1.025} & 0.075 \\ 
      & \textbf{MMD} & \textbf{2.184} & 0.128 & \textbf{2.227} & 0.190 & 0.904 & 0.103 & 1.052 & 0.100 \\ 
      \hline
    \multirow{4}{*}{ $k=0.5$} 
      & \textbf{ES} & \textbf{0.476} & 0.046 & 0.444 & 0.067 & 0.225 & 0.014 & 0.270 & 0.015 \\
      & \textbf{KL} & 0.550 & 0.059 & \textbf{0.498} & 0.064 & 0.252 & 0.029 & 0.317 & 0.045 \\  
      & \textbf{WA} & 0.544 & 0.095 & 0.526 & 0.094 & \textbf{0.189} & 0.035 & \textbf{0.238} & 0.046 \\ 
      & MMD & 0.691 & 0.056 & 0.621 & 0.072 & 0.380 & 0.041 & 0.502 & 0.070 \\ 
      \hline
    \multirow{4}{*}{$\rho=-0.3$} 
      & ES & -0.163 & 0.047 & -0.178 & 0.069 & 0.197 & 0.032 & 0.246 & 0.034 \\ 
      & \textbf{KL} & \textbf{-0.291} & 0.034 & -0.324 & 0.037 & \textbf{0.117} & 0.014 & \textbf{0.144} & 0.020 \\ 
      & \textbf{WA} & -0.288 & 0.026 & \textbf{-0.314} & 0.035 & 0.125 & 0.016 & 0.152 & 0.020 \\
      & MMD & -0.194 & 0.047 & -0.210 & 0.063 & 0.174 & 0.030 & 0.218 & 0.035 \\
      \bottomrule
    \end{tabular}
    }
\end{table*}%

\begin{figure*}[t]
\centering    
    \begin{minipage}[b]{.45\linewidth} 
        \centering    
        \includegraphics[width = .8\textwidth]{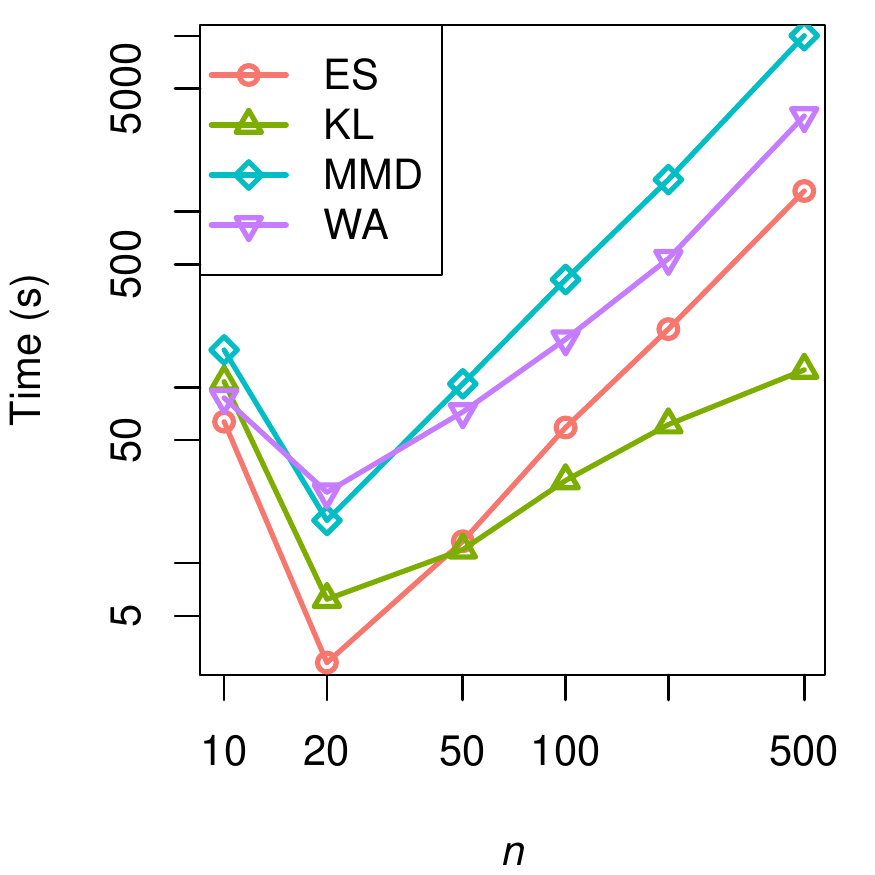}
        \subcaption{Bivariate Gaussian mixtures (Section~\ref{sec:BGM})}\label{fig:1a} 
    \end{minipage}%
    \begin{minipage}[b]{.45\linewidth} 
        \centering    
        \includegraphics[width = .8\textwidth]{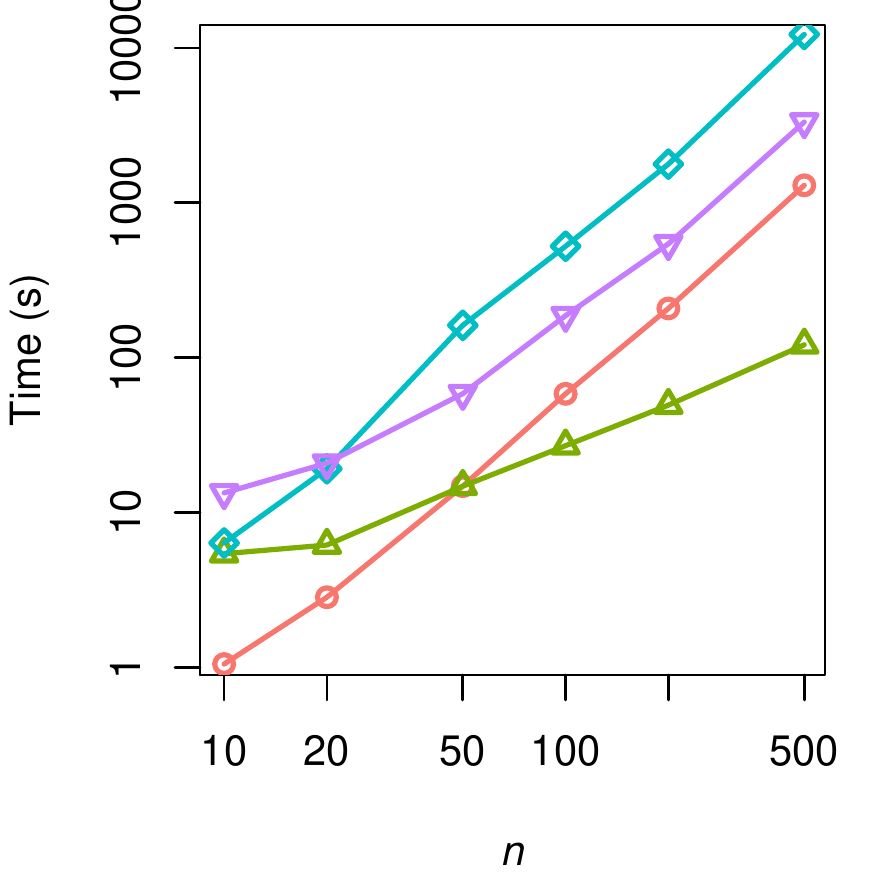}
        \subcaption{MA($2$) model  (Section~\ref{sec:MA2})}\label{fig:1b} 
    \end{minipage} \\
    \vspace{.5cm}
    \begin{minipage}[b]{.45\linewidth} 
        \centering    
        \includegraphics[width = .8\textwidth]{figs/times-3.pdf}
        \subcaption{Bivariate beta model (Section~\ref{sec:BBM})}\label{fig:1c} 
    \end{minipage}%
    \begin{minipage}[b]{.45\linewidth} 
        \centering    
        \includegraphics[width = .8\textwidth]{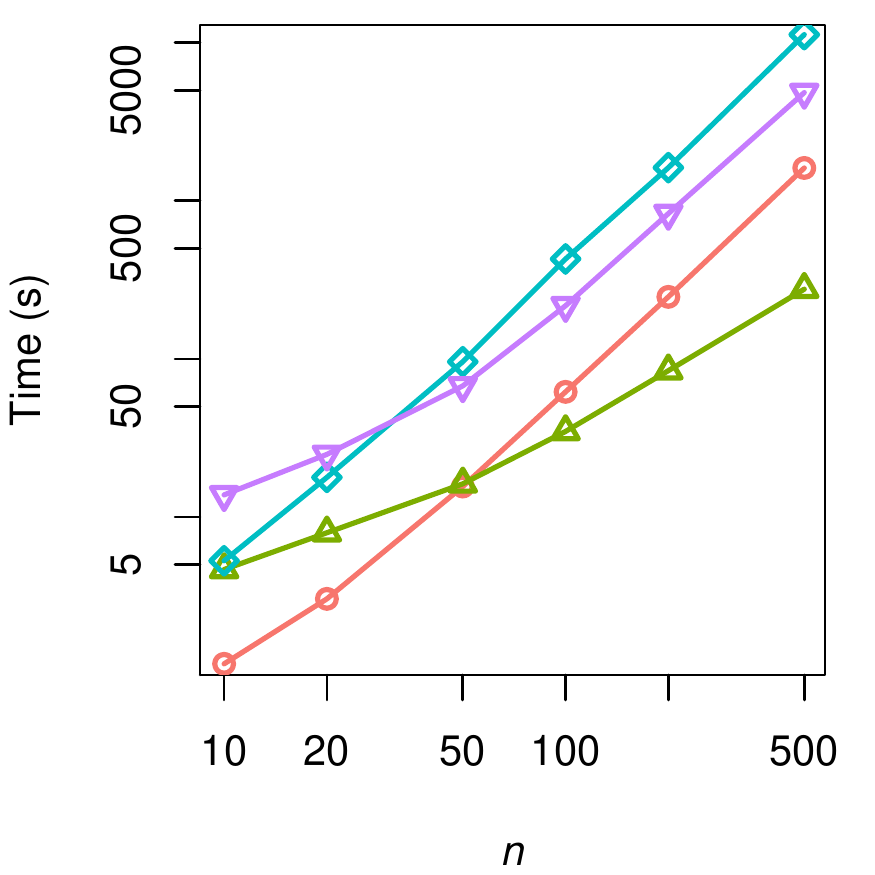}
        \subcaption{$g$-and-$k$ distribution (Section~\ref{sec:gandk})}\label{fig:1d} 
    \end{minipage}%
    \caption{Running times for the ES, KL, MMD and WA distance computations in the four models considered in the simulations. Time in seconds for $10^5$ ABC replications, for varying sample sizes $n$, and with $m=n$. Log scales on both axes.}
    \label{fig:times}
\end{figure*}

\section{Conclusion} \label{sec:conclusion}

We have introduced a novel importance-sampling  ABC algorithm that is based on the so-called \textit{two-sample energy statistic}. Along with other data discrepancy measures that view data sets as empirical measures, such as the Kullback--Leibler divergence, the Wasserstein distance and maximum mean discrepancies, our proposed approach bypasses the cumbersome use of summary statistics. 

We have shown that the V-statistic estimator of the ES is consistent under mild moment conditions. 
Furthermore, we have established a new asymptotic result for cases when the observed sample and simulated sample sizes increasing to infinity, that shows a kind of consistency of the pseudo-posterior in the infinite data scenario. This is in concordance with previous results in such cases \citep[see for instance][]{Jiang2018,Bernton:2017aa} and extends upon existing theory for the application in the general IS-ABC framework. \textcolor{black}{That is, we largely extend the main result of \cite{Jiang2018}, regarding the large sample properties of the pseudo-posterior PDF, to the IS-ABC cases that are considered in \cite{Karabatsos2018} and \cite{Park2016}. Thus, we provide further theoretical justification for the usage of such algorithms.}

Illustrations of the proposed ES-ABC algorithm on four experimental models have shown that it performs comparatively well to alternative discrepancy measures.

\textcolor{black}{Considering computing costs, the ES, KL, MMD, and WA estimators in \textit{univariate settings} are all equal in terms of order of complexity, with a \emph{linearithmic} computational time of $\mathcal{O}((n+m)\log(n+m))$ (see \cite{huo2016fast,chaudhuri2019fast}, regarding the complexity of the ES and MMD estimators). 
In \textit{multivariate settings}, KL complexity is unchanged; ES and MMD have quadratic time $\mathcal{O}((n+m)^2)$, while the Wasserstein distance has complexity $\mathcal{O}((n+m)^{5/2}\log(n+m))$. The latter can be reduced to quadratic complexity if one is targetting the swapping distance, an approximation of the actual Wasserstein distance \citep{Bernton:2017aa}. We note that linear time estimators are also available for the MMD and the ES, if one is willing to forgo precision in the estimates (see \cite{Gretton:2012aa}). See Table~\ref{tab:comput_complexities} for a summary.}


\bibliographystyle{apalike}
\bibliography{ABC_Manuscript}

\end{document}